\documentclass{article}

\usepackage[asymmetric,left=20mm,right=20mm,bottom=36mm,top=33mm]{geometry}
\usepackage{amsthm}
\usepackage{logic9-thm}
\usepackage{tikz}
\usepackage{todonotes}
\usepackage{extarrows}
\usetikzlibrary{decorations.pathmorphing}
\usepackage{wrapfig}

\usepackage{enumerate}

\newtheorem{theorem}{Theorem}
\newtheorem{lemma}[theorem]{Lemma}

\theoremstyle{plain}
\newtheorem{definition}[theorem]{Definition}
 \newtheorem{example}[theorem]{Example}
\newtheorem{fact}[theorem]{Fact}
\newcommand{\Dwords}{\S[\Dd]^*}
\newcommand{\bind}[1]{\downarrow #1}

\title{Defining relations on graphs: how hard is it in the presence of
node partitions?}

\author{
M Praveen and B Srivathsan\\
       Chennai Mathematical Institute\\
       Chennai, India\\
}
\date{}

\begin{document}
\newcommand{\lang}{\mathcal{L}}

\newcommand{\Oh}{\mathcal{O}}

\newcommand{\thmref}[1]{Theorem~\ref{#1}}
\newcommand{\figref}[1]{Fig.~\ref{#1}}
\newcommand{\lemref}[1]{Lemma~\ref{#1}}
\newcommand{\factref}[1]{Fact~\ref{#1}}
\newcommand{\defref}[1]{Definition~\ref{#1}}
\newcommand{\exref}[1]{Example~\ref{#1}}
\newcommand{\secref}[1]{Section~\ref{#1}}
\newcommand{\xra}{\xrightarrow}

\newcommand{\rdpqm}{\mathrm{RDPQ_{mem}}}
\newcommand{\rdpqe}{\mathrm{RDPQ_{=}}}

\newcommand{\botk}{\bot^k}

\usetikzlibrary{shapes.multipart}
\usetikzlibrary{decorations.markings}
\newlength{\ml}
\setlength{\ml}{1cm}

\tikzstyle{state}=[draw=black, circle, inner sep=0.01\ml, minimum
height=0.4\ml]
\tikzstyle{gadget}=[draw=black, rectangle, inner sep=0.01\ml, minimum
height=0.3\ml, minimum width=0.3\ml]

\maketitle
\begin{abstract}
  Designing query languages for graph structured data is an active
  field of research. Evaluating a query on a graph results in a
  relation on the set of its nodes. In other words, a query is a
  mechanism for defining relations on a graph. Some relations may not
  be definable by any query in a given language. This leads to the
  following question: given a graph, a query language and a relation
  on the graph, does
  there exist a query in the language that defines the 
  relation? This is called the definability problem. When the given
  query language is standard regular expressions, the definability
  problem is known to be \PSPACE{}-complete.

  The model of graphs can be extended by labeling nodes with values
  from an infinite domain. These labels induce a partition on the set
  of nodes: two nodes are equivalent if they are labeled by the same
  value. Query languages can also be extended to make use of this
  equivalence. Two such extensions are Regular Expressions with Memory
  (REM) and Regular Expressions with Equality (REE).

  In this paper, we study the complexity of the definability problem
  in this extended model when the query language is either REM or
  REE. We show that the definability problem is \EXPSPACE{}-complete
  when the query language is REM, and it is \PSPACE{}-complete when
  the query language is REE. In addition, when the query language is a
  union of conjunctive queries based on REM or REE, we show
  \coNP{}-completeness.

\end{abstract}

%\category{H.4}{Information Systems Applications}{Miscellaneous}

\section{Introduction}
Graph structures representing data have found many applications like
semantic web \cite{PAG2009,GHM2011}, social networks
\cite{RS2009} and biological networks
\cite{L2005}. One model of graph structured data consists of a set of
nodes labeled by values from some infinite domain and directed edges
between the nodes labeled by letters from a finite alphabet. For
example, a graph representing a social network may have a node for
each member. There may be directed edges labeled $\mathit{friend}$
between two nodes if the corresponding members are friends in the
network. Nodes could be labeled by the name of the corresponding
member's favourite movie. These labels from the infinite domain
partition the set of nodes of the graph. Two nodes are equivalent if
they have the same label. An active field of research is designing
languages for querying such graphs, using both the structure of the
graph and the partition induced by labels from the infinite domain
\cite{LV2012,B2013}.

We will use the term \emph{data graphs} for the model where nodes carry
labels from an infinite domain (a nomenclature from \cite{LV2012}).
The labels themselves are called \emph{data values}. One way of querying data
graphs is to simply specify a language $L$ of strings.  Each string in
$L$ has data values in odd positions and a letter from the finite
alphabet in even positions.  Evaluating a query specified by such a
language on a data graph returns the set of all pairs of nodes
$\struct{q_{1}, q_{2}}$ such that there is a path from $q_{1}$ to $q_{2}$
labeled by a string in the specified language. Register automata
\cite{K1994,IS2000,NSV2004} are extensions of standard finite state
automata for handling data values from infinite domains. Using
register automata as the formalism to specify languages, Libkin and
Vrgo\v{c} studied the complexity of evaluating queries
on data graphs~\cite{LV2012}. There, the main reason behind the choice of
register automata over other formalisms is to obtain tractable
complexity for the query evaluation problem. Aiming towards a
practically usable query language, extensions of standard regular
expressions were introduced in \cite{LV2012}. They are named regular
expressions with memory (REM) and (less expressive) regular
expressions with equality (REE). REM are expressively equivalent to
register automata \cite{LV2012b}. The complexity of query containment
for these have also been studied \cite{KRV2014}.

Here we study the complexity of the definability problem: given a
data graph and a set of pairs of nodes, check if the set can be
obtained as the evaluation of some query on the data graph. One of the
motivations for this study is the extraction of schema mappings, which
we illustrate by an example. Given a data graph representing a social
network, suppose we want to create another graph where two nodes are
in the $\mathit{movieLink}$ relation if they represent people having
the same favourite movie and who are linked by a series of friends.
There is a correspondence between the two graphs; in general such
correspondences are called schema mappings. This particular schema
mapping is specified by saying that the relation $\mathit{movieLink}$
is exactly the relation returned by evaluating the query
$\mathit{friend}^{*}$ on the original graph, with the additional
condition that the two nodes have the same data value (i.e., the same
favourite movie). Given the original data graph and the relation
$\mathit{movieLink}$, suppose we want to algorithmically build the
specification of the schema mapping using some query language. Then we
need to check if the query language is capable of defining the
$\mathit{movieLink}$ relation --- this is the definability problem.
Using example instances of source and target schemas for deriving
appropriate source-to-target mappings have been explored in relational
databases \cite{FGPV2009,GS2010,DPGW2010,ATKT2011}. Research on schema
mappings for graph databases has started \cite{CDLV2011,BPR2013},
though data values and extraction from example graphs have not been
considered till now to the best of our knowledge. Example instances
have also been used to derive ``wrapper'' queries for extraction of
relevant information from data sources \cite{GKBHF2004}.

\textbf{Contributions} We study the complexity of the definability
problem in data graphs, using either REM or REE as query languages. We
prove the following results.
\begin{enumerate}
    \item The definability problem with REM as the query language
        is \EXPSPACE{}-complete.
    \item The definability problem for REM with $k$ memory locations is in
        {\text{\sc Space}}$(\Oh(n\d^{k}))$, where $n$ is the number of
        nodes and $\d$ is the number of data
        values used in the data graph.
    \item The definability problem for REE is \PSPACE{}-complete.
    \item The definability problem for union of conjunctive queries
        based on REM or REE is \coNP{}-complete.
\end{enumerate}
For the upper bounds, we have to overcome some challenges.  In the
presence of data values, standard language theoretic tools like
complementation, determinization and decidability of language
inclusion do not work. We have to understand how data values affect
definability, so that we can appeal directly to the more fundamental
idea of pumping lemma, which still works in the presence of data
values. For the lower bounds, we identify how small data graphs can
count exponentially large numbers using data values, which otherwise
require exponentially large graphs.

\textbf{Related work} Apart from derivation of mappings
\cite{FGPV2009,GS2010,DPGW2010,ATKT2011}, studies have also been made
of using data examples to illustrate the semantics of schema mappings
\cite{ACKT2011}. In \cite{CDK2013}, the problem of deriving schema
mappings from data examples is studied from the perspective of
algorithmic learning theory.

In \cite{ANS2013}, the complexity of the definability problem for
graph query languages is studied, but they do not consider data
values. Their main result is that definability using regular
expressions in \PSPACE{}-complete. They also give upper and lower
bounds for various fragments of conjunctive queries based on regular
expressions. We do not study conjunctive queries or their fragments
but instead give tight bounds for union of conjunctive queries, which
also apply to the setting of \cite{ANS2013} where there are no data
values.

The problem of query containment for fragments and extensions of REM
and REE have been studied in \cite{KRV2014}. A query $e_{1}$ is
contained in another query $e_{2}$ if the set defined by $e_{1}$ is a
subset of the set defined by $e_{2}$ on all data graphs. It is shown
in \cite{KRV2014} that for some fragments of REM and REE, query
containment is respectively \EXPSPACE{}-complete and
\PSPACE{}-complete. These are similar to the bounds we get for the
definability problem. However, the upper bounds in \cite{KRV2014}
apply only to the positive fragments of REM and REE, where tests for
inequality of data values are not allowed (query containment in the
general case is undecidable). There is no obvious way to use those
techniques here, since we allow the full syntax for REM and REE. For
the \EXPSPACE{} lower bound, the authors of \cite{KRV2014} use
techniques similar to those used in \cite{BRL2013} to prove
\EXPSPACE{} lower bound for checking the emptiness of parameterized
regular expressions, closely related to REM. The \EXPSPACE{} lower
bound in \cite{BRL2013} is based on succinctly reducing the emptiness
of intersection of several expressions to the emptiness of a single
expression. Here, we need to use a different approach, since we deal
with the definability problem and can not rely on intersections.

%%% Local Variables: 
%%% mode: latex
%%% TeX-master: "main"
%%% End: 

\section{Preliminaries}
\label{sec:preliminaries}

We will recall the basic definitions. The model of
graphs with node labels from an infinite domain are called data graphs
in~\cite{LV2012}. We will follow the same nomenclature here. We will
also make use of many other notations from \cite{LV2012}.

Let $\S$ be a finite alphabet and let $\Dd$ be a countably infinite
set of data values. We write $[n]$ for the set $\{0, 1, \dots, n \}$.

% Data graph
\begin{definition}[Data graph]
  A \emph{data graph} over $\S$ and $\Dd$ is a triple $G = (V, E,
  \rho)$ where:
  \begin{itemize}
  \item $V$ is a finite set of nodes,
  \item $E \incl V \times \S \times V$ is a set of edges with labels
    in $\S$,
  \item $\r: V \to \Dd$ maps every vertex to a data value.
  \end{itemize}
\end{definition}

\begin{example} Figure~\ref{fig:data-graph} gives an example of a data
  graph over $\S=\{ a \}$ and $\Dd = \mathbb{N}$, the set of natural
  numbers. However, a given graph would use only a finite set of data
  values.  The role of data values will become clearer when we define
  query languages for such data graphs. We will
  use this graph as a running example throughout this section.
   \begin{figure}[h]
    \centering
    \begin{tikzpicture}[vertex/.style={circle, inner sep=2pt,draw,
        thick},pin distance=0.1\ml, scale=0.5]
      \node[state, pin=90:$0$] (v1) at (0,0) {\footnotesize $v_1$};
      \node[state, pin=90:$1$] (v2) at (2,0) {\footnotesize $v_2$};
      \node[state, pin=90:$0$] (v3) at (4,0) {\footnotesize $v_3$};
      \node[state, pin=90:$1$] (v4) at (6,0) {\footnotesize $v_4$};
      \node[state, pin=0:$1$] (u) at (0,-2)
      {\footnotesize $z_2$};
      \node [state, pin=90:$3$] (u2) at (-2,-2)
      {\footnotesize $z_1$};
      \node[state, pin=-90:$2$] (w1) at (0,-4) {\footnotesize $v_1'$};
      \node[state, pin=-90:$3$] (w2) at (2,-4) {\footnotesize $v_2'$};
      \node[state, pin=-90:$2$] (w3) at (4,-4) {\footnotesize $v_3'$};
      \node[state, pin=-90:$3$] (w4) at (6,-4) {\footnotesize $v_4'$};
      
      \begin{scope}[->, >=stealth]
        \draw (v1) -- node [above] {\footnotesize $a$} (v2);
        \draw (v2) -- node [above] {\footnotesize $a$} (v3);
        \draw (v3) -- node [above] {\footnotesize $a$} (v4);
        \draw (v1) -- node [left] {\footnotesize $a$} (u);
        \draw (u) --  node [left] {\footnotesize $a$} (w1);
        \draw (w1) -- node [above] {\footnotesize $a$} (w2);
        \draw (w2) -- node [above] {\footnotesize $a$} (w3);
        \draw (w3) -- node [above] {\footnotesize $a$} (w4);
        \draw (u) -- node [above] {\footnotesize $a$} (v2);
        \draw (w2) -- node [near end, left] {\footnotesize $a$} (v4);
        \draw (v3) -- node [near end, right] {\footnotesize $a$} (w3);
        \draw (u2) -- node [above] {\footnotesize $a$} (u);
      \end{scope}

    \end{tikzpicture}
    \caption{Example of a data graph over a unary alphabet $\Sigma =
      \{ a\}$ and using data values $\{0, 1, 2, 3\}$.}
    \label{fig:data-graph}
  \end{figure}
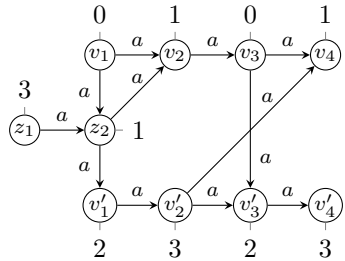
%%% Local Variables: 
%%% mode: latex
%%% TeX-master: "main"
%%% End: 
\end{example}

A \emph{path} in $G$ is a sequence $\xi=v_1 a_1 v_2
a_2 \dots v_{m-1} a_{m-1} v_m$ of nodes in $V$ alternating with
letters in $\Sigma$ such that $(v_i, a_i, v_{i+1})$ is in $E$ for
all $i < m$. The \emph{data path} $w_\xi$ corresponding to a path
$\xi$ is the sequence $\r(v_1) a_1 \r(v_2) a_2 \dots \r(v_{m-1})
a_{m-1} \r(v_m)$ obtained by replacing every node in $\xi$ by its
associated data value. We say that a data path $w$ \emph{connects} node $u$
to $v$ if there is a path $\xi = u a_1 u_1 \dots a_{m-1} v$ in $G$
such that $w_\xi = w$. We write $u \xra{w} v$ in this case.

In general, a data path is a sequence $d_0 a_0 d_1 a_1 \dots a_{m-1}
d_m$ of data values in $\Dd$ alternating with letters in $\Sigma$,
starting and ending with data values. The set of all data paths over
$\S$ and $\Dd$ is denoted by $\S[\Dd]^*$. A \emph{data language} $L
\incl \Dwords$ is a set of data paths. Given two data paths $w_1 = d_0
a_0 d_1 \dots a_{m-1} d_m$ and $w_2 = d_0' b_0 d_1' \dots b_{l-1}
d_l'$ where the last data value of $w_1$ coincides with the first data
value of $w_2$ ($d_m = d_0'$), the
\emph{concatenation} $w_1 \cdot w_2$ is the data path $d_0 a_1 d_1
\dots a_{m-1} d_m b_0 d_1' \dots b_{l-1} d_l'$. This naturally extends
to concatenations of many data paths. We will often write $w_1 w_2$
instead of $w_1 \cdot w_2$.%  If a data path $w = w_1 w_2 \cdots w_l$,
% then we say that it is a \emph{splitting} of $w$ into $w_1, \dots,
% w_l$.

We will now define two formalisms to characterize data
languages. These formalisms would then be used to define query
languages for data graphs. Since $\Dd$ could be infinite, these
formalisms cannot check for the exact data value. They can however
check for equality of two data values.
% Regular expressions with memory
The first formalism is \emph{regular expressions with memory}. They
are extensions of standard regular expressions over the finite
alphabet $\Sigma$, introduced in \cite{LV2012}. They are equipped with
\emph{registers}, that can store data values along a
data path. The stored data values can be used to impose conditions on
the data values allowed in future positions of the data path.

\begin{definition}
  Given a set of registers $r_1, r_2, \dots, r_k$, the set $\Cc_k$ of
  conditions is given by the following grammar:
  \begin{align*}
    c := ~\top~|~r_i^=~|~r_i^{\neq}~|~c \lor c~|~c \land c~|~\neg c, \quad 1
    \le i \le k
  \end{align*}
  The satisfaction is defined with respect to a data value $d \in \Dd$
  and a tuple $\t = (d_1, \dots, d_k) \in (\Dd \cup \bot)^k$ called an
  \emph{assignment}: $d, \t \models \top$ always,  $d, \t  \models r_i^=
    \text{ iff } d_i = d$ and $d, \t  \models r_i^{\neq} \text{ iff }
    d_i \neq d$. 
  % \begin{align*}
  %   d, \t \models \top \text{ always} \quad d, \t  \models r_i^=
  %   \text{ iff } d_i = d  \quad d, \t  \models r_i^{\neq} \text{ iff } d_i \neq d
  % \end{align*}
  The $\bot$ symbol is used to denote an empty register. It satisfies
  $\bot \neq d$ for every data value $d \in \Dd$. Satisfaction for the
  logical operators is as usual.
\end{definition}

% Regular expressions with memory that use registers and the above
% conditions on registers can now be defined.
\begin{definition}[Regular expressions with memory]
  Let $\S$ be a finite alphabet and $r_1, \dots, r_k$ a set of
  registers. Then, \emph{regular expressions with memory (REM)} are
  defined by the following grammar:
  \begin{align*}
    e := \e~|~a~|~e+e~|~e \cdot e~|~e^+~|~e[c]~|~\downarrow \bar{r}.e
  \end{align*}
  where $a\in \Sigma$, $c \in \Cc_k$ and $\bar{r}$ is a tuple of
  registers.
%
%  A regular expression with memory is well-formed if is satisfies two
%  conditions:
%  \begin{itemize}
%  \item Subexpressions $e^+$, $e[c]$ and $\bind{\bar{x}}.e$ are not
%    allowed if $e$ reduces to $\e$. We say that $e$ reduces to $\e$ if
%    it is either $\e$ or it is one of $e_1 + e_2$, $e_1.e_2$, $e_1[c]$
%    or $\bind{\bar{x}}.e_1$ where $e_1$ or $e_2$ reduce to $\e$.
%
%  \item No variable appears in a condition before it appears in
%    $\bind{\bar{x}}.e$.
%  \end{itemize}
%  % The set of well-formed regular expressions with memory is denoted
%  % by
%  % $\REG(\S[x_1,\dots,x_k])$.
\end{definition}

We will use $k$-REM to denote the set of regular expressions with
memory that use at most $k$ registers.  Let $\s$ be an assignment. We will denote by
$\s[\bar{r}\to d]$ the assignment obtained from $\s$ by assigning $d$ to
the registers in $\bar{r}$. The semantics of $k$-REMs are as follows,
reproduced from \cite{LV2012}.
\begin{definition}[Semantics of $k$-REMs] Suppose $e$ is a $k$-REM,
  $w$ is a data path in $\Dwords$ and $\s,\s' \in (\Dd \cup \bot)^k$
  are assignments of the $k$ registers used in $e$. The relation $(e,
  w, \s) \vdash \s'$ is defined by induction on the structure of $e$:
   \begin{align*}
     (\e, w, \s) & \vdash \s' && \text{if $w = d$ for some $d
       \in \Dd$ and $\s = \s'$} \\
     (a, w, \s) & \vdash \s' && \text{if $w = d_1 a d_2$ and $\s' =
       \s$} \\
     (e_1 + e_2, w, \s) & \vdash \s' && \text{if $(e_1, w, \s) \vdash
       \s'$ or $(e_2, w, \s) \vdash \s'$} \\
     (e_1 \cdot e_2, w, \s) & \vdash \s' && \text{if $w = w_1 \cdot
       w_2$ and $\exists~ \s_1 \in (\Dd \cup \bot)^k$ } \\
     &  && \text{ s. t. $(e_1, w_1, \s) \vdash \s_1$ } \\
     & && \text{ and $(e_2, w_2, \s_1) \vdash \s'$}\\
     (e^+, w, \s) & \vdash \s' && \text{if $w =
       w_1 w_2 \dots w_l$
     } \\
     & && \text{and $\exists~ \s_0, \dots, \s_{l} \in (\Dd \cup \bot)^k
       $ s.t.} \\
     & && (e, w_i, \s_{i}) \vdash
     \s_{i+1} \text{ for $i \in [l-1]$} \\
     & && \text{and $\s_0 = \s$, $\s_l = \s'$}\\
     (e[c], w, \s) & \vdash \s' && \text{if $(e, w, \s) \vdash \s'$
       and } \s',d \models c \\
     & && \text{where $d$ is the last data value in $w$} \\
     (\bind{\bar{r}} . e, w, \s) & \vdash \s' && \text{if $(e, w,
       \s[\bar{r}\to d]) \vdash \s'$} \\
     & && \text{where $d$ is the first value in $w$}
   \end{align*}

\end{definition}

The \emph{language} of a $k$-REM $e$ is defined as follows:
\begin{align*}
  \Ll(e) = \{~w\in \Dwords~|~ (e, w, \botk) \vdash \s \text{ for some }
  \s \}
\end{align*}
where $\botk$ denotes the assignment that has $\bot$ in every register.

\begin{example}
  The REM $ {\downarrow} r_1 \cdot a \cdot [r_1^=]$ uses one
  register. The language of this $1$-REM consists of all data paths of
  the form $d a d$ where the first and last data values are the
  same. The $2$-REM ${\downarrow} r_1 \cdot a \cdot {\downarrow} r_2
  \cdot b \cdot a [r_1^=] \cdot b [r_2^{\neq}]$ contains data paths of
  the form $d_1 a d_2 b d_3 a d_4 b d_5$ where $d_1 = d_4$ and $d_2
  \neq d_5$.
\end{example}

The next formalism for characterizing data languages is another
extension of standard regular expressions, called \emph{regular
expressions with equality}, again introduced in \cite{LV2012}. These
are less powerful than REMs, since checking the equivalences between
data values is restricted to a certain form.

\begin{definition}[Regular expressions with equality] Let $\S$ be a
  finite alphabet and $\Dd$ a countably infinite set of data values. A
  \emph{regular expression with equality (REE)} is constructed from
  the following grammar:
  \begin{align*}
    e := \e~|~a~|~e + e~|~e \cdot e~|~e^+~|~e_{=}~|~e_{\neq}
  \end{align*}
  where $a$ belongs to $\Sigma$. The language $\Ll(e)$ of an REE is
  defined as follows:
  \begin{align*}
    \Ll(\e) & ~=~ \{~d~|~d \in \Dd \} \\
    \Ll(a) & ~=~ \{~d_1 a d_2~|~d_1, d_2 \in \Dd \} \\
    \Ll(e_1 + e_2) & ~=~ \Ll(e_1) \cup \Ll(e_2) \\
    \Ll(e_1 \cdot e_2) & ~=~ \Ll(e_1) \cdot \Ll(e_2) \\
    \Ll(e^+) & ~=~ \{ ~w_1 \cdots w_l~| ~l\ge 1 \text{ and each } w_i
    \in
    \Ll(e) ~\} \\
    \Ll(e_=) & ~=~ \{~d_1 a_1 d_2 \dots a_{m-1} d_m \in \Ll(e)~|~ d_1 =
    d_m\} \\
    \Ll(e_{\neq}) & ~=~ \{~d_1 a_1 d_2 \dots a_{m-1} d_m \in \Ll(e)~|~
    d_1 \neq d_m\}
  \end{align*}
\end{definition}

\begin{example}
  The language of the REE $((a)_{\neq} \cdot (b)_{\neq})_{\neq}$
  contains data paths $d_1 a d_2 b d_3$ such that $d_1 \neq d_2$, $d_2
  \neq d_3$ and $d_1 \neq d_3$.
\end{example}

We call a bijection $\pi:\Dd \to \Dd$ an automorphism on $\Dd$, since
it preserves (in)equality.

\begin{definition}
  Let $\pi:\Dd \mapsto \Dd$ be an automorphism on $\Dd$. For a data path
  $w = d_0 a_0 d_1 a_1 \dots d_m$ over $\Sigma[\Dd]^*$, we denote by
  $\pi(w)$ the data path $\pi(d_0) a_0 \pi(d_1) a_1 \dots \pi(d_m)$ obtained by
  applying the automorphism $\pi$ on the data values of $w$.
\end{definition}

An important property of REM and REE is that they cannot distinguish
between automorphic data paths, just like register
automata~\cite{K1994}. 

\begin{fact}[\cite{K1994,LV2012b}]
    \label{fact:autDataPaths}
    For every REM or REE $e$, and for every data path $w \in \Ll(e)$ and automorphism
    $\pi: \Dd \to \Dd$, $\pi(w)$ is also in $\Ll(e)$.
\end{fact}

\subsection{Query languages for data graphs}

The above two formalisms can be used to define query languages for
data graphs. 

\begin{definition}[Regular data path queries]
  An expression $Q=x \xra{e} y$ is a \emph{regular data path query},
  when $e$
  is either a standard regular expression, or an REM or an REE. Given a
  data graph $G$, the result of the query $Q(G)$ is the set of pairs of
  nodes $\struct{u,v}$ such that there exists a data path from $u$ to $v$
  that belongs to $\Ll(e)$. The query is called \emph{regular data
    path query with memory (RDPQ$_{mem}$)} or \emph{regular data path
    query with equality (RDPQ$_=$)} depending on whether $e$ is an REM
  or an REE. If $e$ is a standard regular expression, the query is
  called a \emph{regular path query (RPQ)}. 
  % given a name depending on the
  % formalism used. It is called a:
  % \begin{itemize}
  % \item \emph{regular path query (RPQ)} if $e$ is a standard
  %   regular expression over $\S$,
  % \item \emph{regular data path query with memory (RDPQ$_{mem}$)} if
  %   $e$ is an REM,
  % \item \emph{regular data path query with equality (RDPQ$_=$)} if $e$
  %   is an REE.
  % \end{itemize}
\end{definition}

A \emph{relation} on the set of nodes in the graph is a set of tuples
of same arity. We will say that a relation $S$ on a data graph $G$ is
\emph{defined} by a query $Q$ if $S$ equals $Q(G)$.

\begin{example}\label{eg:queries-REM-REE}
  Evaluating the RPQ $Q_1: x \xra{aaa} y$ on the data graph in
  Figure~\ref{fig:data-graph} results in the relation
    $S_1 ~=~ \{~ \struct{v_1, v_4},\\ \struct{v_1, v_3'}, ~
    \struct{v_1, v_3}, ~ \struct{v_1, v_2'},~
    \struct{v_2,
    v_4'},~
    \struct{z_1, v_3}, \struct{z_1, v_2'}, \struct{z_2, v_4},\\ \struct{z_2, v_3'}, \struct{v_1',
    v_4'}~\}$.
  This is the set of all pairs of nodes connected by
  $aaa$. Neither
    $S_2 = \{\struct{v_1, v_4}, \struct{v_1', v_4'}\}$ nor $S_3  =
    \{\struct{v_1, v_3}\}$ can be defined using RPQs. To see why,
  consider $S_2$. The only path connecting $v_1'$ to $v_4'$ is
  $aaa$. But this path connects many other pairs apart from the ones
  in $S_2$. Hence to restrict to the pairs in $S_2$, we need to make
  use of data values. A similar argument will tell us that to define
  $S_3$, we need to consider data values.

  The relation $S_2$ can be defined by the $\rdpqm$
    $Q_2: x \xra{e_2} y$, where $e_2 = \downarrow r_1
    \cdot a \cdot \downarrow r_2 \cdot a [r_1^=] \cdot a [r_2^=]$.
  The REM $e_2$ contains all data paths $d_1 a d_2 a d_3 a d_4$ such
  that $d_1 = d_3$ and $d_2 = d_4$. From Figure~\ref{fig:data-graph},
  it can be checked that the only data paths in the graph satisfying
  this expression are: $w_1: 0a1a0a1$ and $ w_2: 2a3a2a3$, 
  % \begin{align*}
  %   w_1: 0a1a0a1 \quad \text{ and } \quad w_2: 2a3a2a3
  % \end{align*}
  and they connect $\struct{v_1, v_4}$ and $\struct{v_1', v_4'}$ respectively. Hence
  $Q_2(G) = S_2$, thus defining the relation $S_2$. Note that the two
  words $w_1$ and $w_2$ are automorphic images and hence cannot be
  distinguished by REMs (c.f. \factref{fact:autDataPaths}).

  The expression $e_2$ is a $2$-REM (uses $2$ registers $r_1$ and
  $r_2$). Let us see why $S_2$ cannot be
  defined using a $1$-REM. Suppose $e$ is a $1$-REM used in a query
  defining $S_2$. As the only data path connecting $v_1'$ to $v_4'$ is
  $2a3a2a3$, $\Ll(e)$ should contain the data path $2a3a2a3$.
  Moreover, the data paths $0a1a0a2$ and $1a2a3a2$ should not be in
  $\Ll(e)$, since $v_{1} \xra{0a1a0a2} v_{3}'$ and $z_{2}
  \xra{1a2a3a2} v_{3}'$. Since the prefix of $2a3a2a3$ (to be
  included) and $0a1a0a2$ (to be excluded)
  up to the first three data values are automorphic, the only way to
  add $2a3a2a3$ to $\Ll(e)$ and eliminate $0a1a0a2$ from $\Ll(e)$ is
  to check in $e$ that the second and fourth data values are equal.
  This will still not eliminate $1a2a3a2$. To eliminate $1a2a3a2$, one
  has to add the condition that the first and third data values are
  equal. So we need to compare the first data value to the third, and
  second data value to the fourth. This kind of an ``interleaved''
  check needs $2$ registers as in the REM $e_2$ above. For the same
  reason, $S_2$ cannot be defined using $\rdpqe$.

  The relation $S_3$ can be defined using the
  $\rdpqe$ $Q_3: x \xra{e_3} y$ with $e_3 = (a\cdot (a)_=
    \cdot a)_=$
  % \begin{align*}
  %   Q_3: x \xra{e_3} y \text{ with $e_3$ being } (a\cdot (a)_=
  %   \cdot a)_=
  % \end{align*}
  In the data graph of Figure~\ref{fig:data-graph}, the only
  data path satisfying $e_3$ is $ w_5: 0a1a1a0$.
  % \begin{align*}
  %   w_5: 0a1a1a0
  % \end{align*}
  that connects $\struct{v_1, v_3}$. Both the checks in $e_3$: first data
  value equals fourth data value, and second equals third, are
  required to eliminate the following words: $w_6: 3a1a1a0$ and $w_7:
  1a2a3a1$, 
  % \begin{align*}
  %   w_6: 3a1a1a0 \quad \text{ and } \quad w_7: 1a2a3a1
  % \end{align*}
  that connect $\struct{z_1,v_3}$ and $\struct{z_2, v_4}$ which are not in
  $S_3$. Hence for similar reasons as mentioned in the above
  paragraph, $S_3$ cannot be defined by an $\rdpqm$ that
  uses a $1$-REM. A $2$-REM would work though.
\end{example}

We will also study a standard extension of query languages: union of
conjunctive queries.

\begin{definition}[Conjunctive data path queries]
    \label{def:cdpq}
  A \emph{conjunctive regular data path query (CRDPQ)} is an
  expression of the form
  \begin{align}
    \label{eq:crdpq}
    \mathit{Ans}(\bar{z}) := & \Land_{1 \le i \le m} x_{i}
    \xrightarrow{e_i} y_{i},
  \end{align}
  where $m \ge 0$, $x_{i},y_{i}$ are variables and $\bar{z}$ is a
  tuple of variables among $\bar{x}$ and $\bar{y}$ and either every
  $e_i$ is an REM, or every $e_i$ is an REE. The semantics of a CRDPQ
  $Q$ of the form \eqref{eq:crdpq} over a data graph $G = (V, E,
  \rho)$ is defined as follows. Given a valuation $\mu: \bigcup_{1 \le i
  \le m} \{x_{i}, y_{i}\} \to V$, we write $(G, \mu) \models Q$ if
  $\struct{\mu(x_{i}), \mu(y_{i})}$ is in the answer of $x_{i}\xrightarrow{e_{i}}
  y_{i}$ on $G$, for each $i=1, \ldots, m$.  Then $Q(G)$ is the set of
  all tuples $\mu(\bar{z})$ such that $(G, \mu) \models Q$. The number of
  variables in $\bar{z}$ is the \emph{arity} of $Q$.
  %A CRDPQ with the
  %head $\mathit{Ans}()$(i.e., no variables in the output) is a
  %\emph{Boolean} query, which is true on $G$ if $(G, \mu) \models Q$ for
  %some $\mu$.
  A \emph{union of conjunctive regular data path queries}
  (UCRDPQ) is a finite set $Q = \{Q_{1}, \ldots, Q_{k}\}$ of CRDPQs
  $Q_{1}, \ldots, Q_{k}$, which are all of the same arity. For a data
  graph $G$, $Q(G)$ is the set $Q_{1}(G) \cup \cdots \cup Q_{k}(G)$.
\end{definition}

\begin{example} We will work on the same graph from
  Figure~\ref{fig:data-graph}. Consider the following CRDPQ $Q_4$:
  $\mathit{Ans}(x_1, y_1) := x_1 \xra{a} y_1 ~\land~ x_1 \xra{a} y_2
  ~\land~ y_2 \xra{a} y_1$
  % \begin{align*}
  %   \mathit{Ans}(x_1, y_1) := x_1 \xra{a} y_1 ~\land~ x_1 \xra{a} y_2
  %   ~\land~ y_2 \xra{a} y_1
  % \end{align*}
  The only valuation $\mu$ satisfying the above conditions is:
  $\mu(x_1) = v_1$,  $\mu(y_1) = v_2$ and $\mu(y_2) = z_2$. 
  % \begin{align*}
  %   \mu(x_1) = v_1 \quad \mu(y_1) = v_2 \quad \mu(y_2) = z_2
  % \end{align*}
  The result $Q_4(G)$ would hence be the relation $\{~\struct{v_1,
  v_2}~\}$. Note that this relation cannot be defined using $\rdpqm$
  or $\rdpqe$. The only data paths connecting $\struct{v_1, v_2}$ are $0a1$
  and $0a1a1$. The former data path cannot be used to distinguish
  $\struct{v_1, v_2}$ as it connects $\struct{v_3, v_4}$ as well and the latter one
  cannot be used since an automorphic data path $3a1a1$ connects
  $\struct{z_1, v_2}$. From \factref{fact:autDataPaths}, we know that
  REMs and REEs cannot differentiate between automorphic data paths.

  Consider another query $Q_5$: $\mathit{Ans}(x_1, y_1, x_2) :=
  x_1 \xra{(a)_{\neq}} y_1 ~\land~ 
    x_2 \xra{(a)_{\neq}} y_1$.
  % \begin{align*}
  %   \mathit{Ans}(x_1, y_1, x_2) := x_1 \xra{(a)_{\neq}} y_1 ~\land~
  %   x_2 \xra{(a)_{\neq}} y_1
  % \end{align*}
  The above query uses REEs in its individual regular data path
  queries. The result $Q_5(G)$ would be: $ \{~\struct{v_1, z_2, z_1},
  \struct{v_3, v_4, v_2'}, \struct{v_3, v_3', v_2'}~\}$ 
  % \begin{align*}
  %   \{~\struct{v_1, z_2, z_1}, \struct{v_3, v_4, v_2'}, \struct{v_3, v_3', v_2'}~\}
  % \end{align*}
  The query singles out the ``pattern'' of $x_{1}$ and $x_{2}$
  converging into $y_{1}$, where the label of $y_{1}$ is different
  from those of $x_{1}$ and $x_{2}$.
\end{example}

\subsection{Definability problems}

From the examples, we can infer that $\rdpqm$ and $\rdpqe$
can define more relations than RPQ. In addition,
$\rdpqm$ can define more relations than $\rdpqe$. CRDPQs
can define even more than $\rdpqm$. Restricting to $\rdpqm$, using $k$
registers we can define relations that are not possible with $k-1$
registers. It is also not difficult to construct examples of graphs
and relations that are not definable using any of the query languages
that we have seen. This motivates us to look at the following
definability problems. The input is a data graph $G$ and a relation
$S$ on the set of nodes in $G$.
\begin{align*}
   \text{\textit{$\rdpqm$-definability:}} &~ \text{Does there exist
    an $\rdpqm$
    $Q$ } \\
   &~  \text{s.t. $Q(G) = S$? } \\
   \text{\textit{k-$\rdpqm$-definability:}} &~ \text{Does there exist
    an $\rdpqm$
    $Q$ } \\
   &  ~\text{which uses at most $k$ registers } \\
   & ~ \text{s.t. $Q(G) = S$? } \\
   \text{\textit{$\rdpqe$-definability:}} &~ \text{Does there exist
    an $\rdpqe$
    $Q$ } \\
   & ~ \text{s.t. $Q(G) = S$? } \\
   \text{\textit{UCRDPQ-definability:}} &~ \text{Does there exist a
    UCRDPQ
    $Q$ } \\
   &  ~\text{s.t. $Q(G) = S$? }
\end{align*}

In the subsequent sections, we study the complexity of the above
problems.  For the last problem, we do not make a distinction between
UCRDPQs using REM or REE, as we will
see that the complexity stays the same in both cases.

\textbf{Speciality of the equivalence relation} As stated before,
the data values induce an equivalence relation on the set of nodes,
where two nodes are equivalent when they have the same label. Each
letter $a$ from the finite alphabet $\S$ also induces a binary
relation on the set of nodes: the pair $\struct{u,v}$ is in this
relation when there is an edge labeled $a$ from $u$ to $v$. Given that
data values also induce a binary relation, why is it that we can not
solve the definability problem by simply treating the equivalence
relation as an extra letter in the finite alphabet and using the
techniques developed for RPQs? The reason is that query languages give
a special privilege to the equivalence relation: it can be used to
relate positions that are far apart in a data path, while the binary
relation induced by a letter in the finite alphabet can only relate
successive positions. Hence, as seen in \exref{eg:queries-REM-REE},
some relations that can be defined by $\rdpqm$ can not be defined by
RPQ, even if we add the equivalence relation as an extra letter in the
finite alphabet. However, a more sophisticated extension of the graph
will allow us to use this idea, as explained in the beginning of the
next section.

% Data graphs

% REM

% REE

% Conjunctive queries

% Definability problems

%%% Local Variables: 
%%% mode: latex
%%% TeX-master: "main"
%%% End: 

% TODO:

% Define $[n] = \{0,\ldots, n\}$ in the preliminaries.

% Introduce the acronym REM in preliminaries.

% acronym: RDPQ and REE

% prove automorphic data paths can not be distinguished in
% preliminaries

% \todo{extend conditions
%  to include $\top$}

%\todo{consistently use $\d$ for number of data values.}

%\todo[inline]{formalize data path of a path}

\section{Queries using Regular Expressions with Memory}

In this section we study the $\rdpqm$-definability and the
$k$-$\rdpqm$-definability problems. Fix a data graph $G$ and a binary
relation $S$ on the vertices in $G$. We denote the set of data values
in $G$ by $\Dd_G$. The
goal is to decide if $S$ is $\rdpqm$-definable. We start with some
basic observations about the strengths and weaknesses of REMs.

If $w$ is a data path and $\pi: \Dd_G \to \Dd_G$ is an automorphism on
$\Dd_G$, we have seen in Fact~\ref{fact:autDataPaths} that no REM can
distinguish between $w$ and $\pi(w)$.  On the other hand, if two data
paths are not automorphic, then they can be distinguished by an REM.
We denote by $[w]$ the set of all data paths automorphic to $w$.
\begin{lemma}
    \label{lem:canonicalExpression}
    For every data path $w$, there is an REM $e_{[w]}$ such that
    $\Ll(e_{[w]}) = [w]$.
\end{lemma}
\begin{proof}
    Suppose $d_{1}, \ldots, d_{k}$ are the distinct data values
    occurring in $w$. The required REM $e_{[w]}$ uses $k$ registers
    $r_{1}, \ldots, r_{k}$. Essentially, at the first position where
    the data value $d_{i}$ appears, $e_{[w]}$ stores $d_{i}$ in the
    register $r_{i}$. In every subsequent position where $d_{i}$
    appears, it is compared against the value stored in $r_{i}$.
    Formally, $e_{[w]}$ is defined as follows by induction on length
    of $w$: $e_{[d_{i}]} \quad = \quad \downarrow r_{i}$.
    \begin{align*}
        e_{[wad_{i}]} \quad &= \quad
        \begin{cases}
            e_{[w]} \cdot a[r_{i}^{=}] & \text{if } d_{i} \text{ occurs in } w\\
            e_{[w]} \cdot a\cdot \downarrow r_{i}.\e & \text{otherwise}
        \end{cases}
    \end{align*}
    For every data value $d_{i}$ occurring in $w$, $e_{[w]}$ notes all
    the positions having the data value $d_{i}$. Using this, it is
    routine to prove that any data path $w'$ is in $\Ll(e_{[w]})$ iff
    $w'$ is automorphic to $w$.
\end{proof}

Combining Fact~\ref{fact:autDataPaths} and
\lemref{lem:canonicalExpression}, we infer that two data paths can be
distinguished by an REM iff they are not automorphic. This suggests
the following procedure for checking $\rdpqm$-definability. For
simplicity, suppose that we want to define the singleton set
$\{\struct{u,v}\}$. Suppose there is a data path $w$ connecting $u$ to
$v$. The expression $e_{[w]}$ will not define the set
$\{\struct{u,v}\}$ iff there is an automorphism $\pi$ such that
$\pi(w)$ connects $u'$ to $v'$ for some $\struct{u',v'} \ne
\struct{u,v}$. The automorphism $\pi$ is obstructing $e_{[w]}$ from
defining $\{\struct{u,v}\}$, but this obstruction is not explicit in
the data graph $G$. It \emph{is} explicit in $G_{\pi^{-1}}$ (obtained
from $G$ after replacing every data value $d$ by $\pi^{-1}(d)$), since
$w$ connects $u'$ to $v'$ in $G_{\pi^{-1}}$. All such obstructions
will be explicit in $G_{\mathit{aut}}$, the disjoint union of
$G_{\pi}$ for all automorphisms $\pi$. A little more work will allow
us to drop the special treatment given to data values and treat them
as usual letters from a finite alphabet in $G_{\mathit{aut}}$. The
$\rdpqm$-definability problem on $G$ can be reduced to the
RPQ-definability problem on $G_{\mathit{aut}}$. The
\PSPACE{}-completeness of RPQ-definability \cite{ANS2013} will then
give an \EXPSPACE{} upper bound for $\rdpqm$-definability.  This
approach however does not throw light on the role of registers in
definability, nor does it give precise bounds in the case where the
number of registers is fixed.

In the next sub-section, we make some observations on
$k-\rdpqm$-definability, which are counterparts of the above
observations on $\rdpqm$-definability.

\subsection{$\rdpqm$-definability with bounded number of registers}

If $w$ is a data path with $k$ distinct data values, we saw in
\lemref{lem:canonicalExpression} that there is a REM $e_{[w]}$ whose
language is exactly $[w]$. The number of registers used in $e_{[w]}$
is $k$. If we restrict the number of registers to less than $k$, then
there may not be an expression whose language is exactly $[w]$. Still,
the expression $e_{[w]}$ (which uses $k$ registers) has a simple
syntactic form, which we would like to capture and use in scenarios
where there are fewer registers.
\begin{definition}[Basic REM]
  A \emph{basic $k$-REM} is a $k$-REM of the form $\downarrow
  \overline{r}_{1}.a_{1}[c_{1}] \cdot \downarrow
  \overline{r}_{2}.a_{2}[c_{2}] \cdots \downarrow
  \overline{r}_{m}.a_{m}[c_{m}]$, where $a_{i} \in \Sigma$,
  $c_{i} \in \Cc_k$ and $\overline{r}_{i}$ are tuples from
  $r_1, \dots, r_k$.
\end{definition}

Basic $k$-REMs can also be thought of as those built without using the
rules $e := e^{+}$ and $e := e + e$. We considered defining a
singleton set $\{\struct{u,v}\}$ for simplicity. We would like to
retain the simplicity but handle arbitrary sets, which is the purpose
of the following definition.
\begin{definition}
    \label{def:definabilityWitness}
    Suppose $G$ is a data graph, $S$ is a binary relation on the set
    of nodes of $G$ and $\struct{u,v} \in S$. A \emph{$k$-REM witness
    for $\struct{u,v}$ in $S$} is a basic $k$-REM $e$ satisfying the
    following conditions.
    \begin{enumerate}
        \item (Connecting path) $u \xra{w} v$ for some $w \in \Ll(e)$.
        \item (No extraneous pairs) If any data path in $\Ll(e)$
            connects some $u'$ to some $v'$, then
            $\struct{u',v'} \in S$.
    \end{enumerate}
\end{definition}

If an arbitrary $k$-REM $e$ defines $S$ and
$\struct{u,v} \in S$, then there is a data path $w \in \Ll(e)$
connecting $u$ to $v$. If $e$ is of the form $e_{1}e_{2}^{+}e_{3}$,
there is an $m$ such that $\Ll(e_{1}e_{2}^{m}e_{3})$ contains $w$.
Continuing this process of removing iterations in $e$, while still
retaining $w$ in the language will result in a
$k$-REM witness for $\struct{u,v}$ in $S$.
\begin{lemma}\label{lem:union-of-basic}
  If $S$ is definable, then it is definable by a union of $|S|$
  $k$-REM witnesses.
\end{lemma}
\begin{proof}
    Suppose a $k$-REM $e$ defines $S$ and $\struct{u,v} \in S$. We
    will show that there exists a $k$-REM witness for
    $\struct{u,v}$ in $S$, which will prove the lemma.

    Since $e$ defines $S$, there is a data path $w \in \Ll(e)$
    connecting $u$ to $v$. Without loss of generality, we can assume
    that $e$ is of the form $e_1 + e_2 + \dots + e_m$ for some $m \ge
    1$ such that each $e_i$ is union-free, that is each $e_i$ is
    constructed using the grammar for REMs without the $e := e + e$
    rule. The data path $w$ belongs to $\Ll(e_i)$ for some $i \in \{1,
    \dots, m\}$. If $e_i$ is a basic $k$-REM, then we are done.
    Otherwise, $e_i$ is of the form
    $f_1 \cdot (f_2)^+ \cdot f_3$ where $f_1, f_2$ and $f_3$ are union free
    $k$-REMs. As $w \in \Ll(e_i)$, from the semantics of $k$-REMs,
    there exists a number $\a$ such that $w$ satisfies the $k$-REM
    obtained by $\a$ iterations of $f_2$. More precisely, there exist
    data paths $x \in \Ll(f_{1})$, $y_1, \dots, y_\a \in
    \Ll(f_{2})$ and $z \in \Ll(f_{3})$ such that $w = x\cdot y_1\cdots
    y_\a \cdot z$.

    Let us write $(f_2)^\a$ for the $k$-REM obtained by concatenating
    $f_2$ $\a$ times.  If $f_1 (f_2)^\a f_3$ is basic, then we are done.
    Otherwise continue this process of ``unfolding'' to get a basic
    $k$-REM $e'$. Since $w$ is in $\Ll(e')$, $e'$ satisfies the first
    condition of \defref{def:definabilityWitness} (connecting path).
    Since $\Ll(e') \subseteq \Ll(e)$ and $e$ defines $S$, $e'$
    satisfies the second condition of \defref{def:definabilityWitness}
    (no extraneous pairs). Hence, $e'$ is a $k$-REM witness for
    $\struct{u,v}$ in $S$, which finishes the proof.
\end{proof}

Now suppose we are trying to define $S$ using $k-\rdpqm$ and let 
$\struct{u,v} \in S$. Assume there is a data path $w$ connecting $u$ to $v$
and there is a basic $k$-REM $e$ such that $w \in \Ll(e)$. If
$e$ is not a $k$-REM witness for $\struct{u,v}$ in $S$, then there is
a data path $w' \in \Ll(e)$ connecting $u'$ to $v'$ for some
$\struct{u',v'} \ne \struct{u,v}$. The data path $w'$ is obstructing
$e$ from being a witness and we need a structure where such
obstructions are explicit. Since we are dealing with $k$-REMs, the
structure would have to keep track of possible values stored in the
$k$ registers. The following definition and lemma are similar to the way
the semantics of REM over a data graph is defined in \cite{KRV2014}.
\begin{definition}[Assignment graph]
  Let $k$ be a natural number. To a data graph $G = (V, E, \rho)$ over
  finite alphabet $\Sigma$ and data values $\Dd_G$ we associate a
  transition system $\Tt_G = (Q_G, \to_G)$ called the
  \emph{$k$-assignment graph}. Its set of states is $Q_{G} = V
  \times (\Dd_G \cup \bot)^k$. The transitions are of the form
  $\downarrow \overline{r}.a[c]$, where $\overline{r}$ is a (possibly
  empty) tuple of variables from $r_{1}, \ldots, r_{k}$, $a \in
  \Sigma$ and $c$ is a condition in $\Cc_{k}$. There is a transition $(v,\s)
  \xra{~\downarrow\overline{r}.a[c]~}_{G} (v', \s')$ if $(v,a,v') \in
  E$, $\s' = \s[\overline{r} \to \rho(v)]$ and $\rho(v'), \s' \models
  c$.
\end{definition}
A sequence of the form $(v_{0}, \s_{0})
\xra{\downarrow\overline{r}_{1}.a_{1}[c_{1}]}_{G} (v_{1}, \s_{1})
\xra{}_{G} \cdots \xra{\downarrow\overline{r}_{m}.a_{m}[c_{m}]}_{G}
(v_{m},\s_{m})$ in $\Tt_G$ is called a run from $(v_{0},\s_{0})$ to
$(v_{m}, \s_{m})$. The sequence $\downarrow
\overline{r}_{1}.a_{1}[c_{1}] \cdots \downarrow
\overline{r}_{m}.a_{m}[c_{m}]$ is a basic $k$-REM.  Hence, we can
think of runs in $\Tt_{G}$ as being of the form $(u,\s) \xra{~e~}_{G}
(v,\s')$, where $e$ is the basic $k$-REM formed by the labels of the
sequence of transitions connecting $(u,\s)$ to $(v,\s')$.  This
observation leads to the following connection between runs in the
assignment graph and data paths in $G$ belonging to the languages of
basic REMs.

\begin{lemma}
  \label{lem:confGraphREMLang}
  Let $e$ be a basic $k$-REM. Let $\s: \{r_{1}, \ldots, r_{k}\} \to
  \Dd_G \cup \{\bot\}$ and $\s':\{r_{1}, \ldots, r_{k}\} \to \Dd_G
  \cup \{\bot\}$ be some assignments. The following are equivalent.\\
  \hphantom{a}1. A data path $w$ connects $u$ to $v$ in $G$ and $(e,w,\s)
    \vdash \s'$.\\
  \hphantom{a}2. There exists a run $(u, \s ) \xra{~e~}_G (v, \s')$ in $\Tt_G$.
\end{lemma}
\begin{proof}
    By an induction on the number of blocks of the form $\downarrow
    \overline{r}.a[c]$ in $e$. Suppose $e = \downarrow
    \overline{r}.a[c]$. If there is a data path $w$ as in the lemma,
    then $(u,a,v) \in E$ and by the semantics of REMs, $\s' =
    \s[\overline{r} \to \rho(u)]$ and $\rho(v), \s'
    \models c$. Hence $(u,\s)
    \xra{~\downarrow \overline{r}.a[c]~}_{G} (v,\s')$. Conversely,
    let us suppose $(u, \s ) \xra{~\downarrow \overline{r}.a[c]~}_G (v, \s')$
    in $\Tt_G$. By definition of $\Tt_{G}$, we have $(u,a,v) \in E$,
    $\s' = \s[\overline{r} \to \rho(u)]$ and $\rho(v), \s' \models
    c$. Hence, $\rho(u) a \rho(v)$ is a data path connecting $u$ to
    $v$ in $G$ and $(\downarrow \overline{r}.a[c],\rho(u) a
    \rho(v),\s) \vdash \s'$.

    For the induction step, suppose $e = \downarrow \overline{r}.a[c]
    \cdot e'$. If there is a data path $w$ as in the lemma, then $w =
    \rho(u) a \rho(u_{0}) \cdot w'$ for some node $u_{0}$ and a data
    path $w'$ connecting $u_{0}$ to $v$ in $G$. In addition, by the
    semantics of REMs, there is some assignment $\s_{0}$ such that
    $(\downarrow \overline{r}.a[c], \rho(u) a \rho(u_{0}),\s) \vdash
    \s_{0}$ and $(e',w',\s_{0}) \vdash \s'$. Now we can use an
    argument similar to the one in the base case to infer that $(u,\s)
    \xra{~\downarrow \overline{r}.a[c]~}_{G} (u_{0},\s_{0})$ and use
    the induction hypothesis to infer that $(u_{0}, \s_{0})
    \xra{~e'~}_{G} (v,\s')$. Hence, $(u, \s) \xra{~e~}_{G} (v,\s')$.
    Conversely, suppose $(u, \s) \xra{~e~}_{G} (v,\s')$ in $\Tt_{G}$.
    This run can be split as follows: $(u, \s) \xra{~\downarrow
    \overline{r}.a[c]~} (u_{0}, \s_{0}) \xra{~e'~}_{G} (v,\s')$
    for some node $u_{0}$ and assignment $\s_{0}$. Then we can
    argue as in the base case to infer that the data path $\rho(u)
    a \rho(u_{0})$ connects $u$ to $u_{0}$ in $G$ and $(\downarrow
    \overline{r}.a[c], \rho(u) a \rho(u_{0}), \s) \vdash \s_{0}$.
    We can use the induction hypothesis to infer that there is a
    data path $w'$ connecting $u_{0}$ to $v$ in $G$ and $(e',
    w',\s_{0}) \models \s'$. Hence, the data path $\rho(u) a
    \rho(u_{0}) \cdot w'$ connects $u$ to $v$ in $G$ and $(e,
    \rho(u) a \rho(u_{0}) \cdot w', \s) \vdash \s'$.
\end{proof}

Suppose we are trying to define a set $S$ on the data graph $G$ using
$k-\rdpqm$. The above lemma allows us to think of $k$-REM witnesses in
terms of runs in $\Tt_{G}$. A basic $k$-REM $e$ is a $k$-REM witness
for $\struct{u,v}$ in $S$ iff it satisfies the following conditions.
\begin{enumerate}
    \item $(u,\bot^{k}) \xra{~e~}_{G} (v,\s)$ for some assignment
        $\s$, to satisfy condition 1 of
        \defref{def:definabilityWitness} (connecting path).
    \item If $(u',\bot^{k}) \xra{~e~}_{G} (v',\s)$ for some nodes
        $u',v'$ and some assignment $\s$, then
        $\struct{u',v'} \in S$, to satisfy condition 2
        of \defref{def:definabilityWitness} (no extraneous pairs).
\end{enumerate}
Checking that a basic $k$-REM $e$ is a witness thus reduces to
checking that $e$ connects a pair in $\Tt_{G}$ and does not
connect certain other pairs. This observation allows us to use the
pigeon hole principle to prove the existence of short witnesses.

\begin{lemma}
  \label{lem:small-witness}
  Suppose $G$ is a data graph with $\d$ distinct data values, $n$
  nodes $v_{1}, \ldots, v_{n}$ and $S$ is a binary relation on the set
  of nodes. If there is a $k$-REM witness for $\struct{v_{p}, v_{q}}$
  in $S$, there is one of length $\Oo\left(2^{n^2 \d^k}\right)$.
\end{lemma}
\begin{proof}
  For sets of states $Q_{1}, \ldots, Q_{n}, Q_{1}', \ldots, Q_{n}'
  \incl Q_{G}$, we write $\struct{Q_{1}, \ldots, Q_{n}} \xra{~e~}_{G}
  \struct{Q_{1}', \ldots, Q_{n}'}$ if
    $Q_{i}' = \{(v',\s') \mid (v,\s) \xra{~e~}_{G} (v',\s') \text{ for
    some } (v,\s) \in Q_{i}\}$ for every $i=1, \ldots, n$.
  Suppose $e$ is a $k$-REM witness for $\struct{v_{p}, v_{q}}$ in
  $S$. Let $e = e_{1} \cdot e_{2} \cdots e_{m}$, where every $e_{i}$
  is of the form $\downarrow \overline{r}_{i} a_{i}.[c_{i}]$. Consider
  the sequence:
  \begin{align}\label{eq:tuple-sequence-in-configuration-graph}
    \struct{\{(v_{1}, \bot^{k})\}, \ldots, \{(v_{n}, \bot^{k})\}}
    &  \xra{e_{1}}_{G} \struct{Q_{1}^{1}, \ldots, Q_{n}^{1}} \\
    \nonumber & \xra{e_{2}}_{G} \cdots \xra{e_{m}}_{G} \struct{Q_{1}^{m},
      \ldots, Q_{n}^{m}} \enspace .
  \end{align}
  The set $Q_{i}^{j}$ is the set of all states reachable from
  $(v_{i},\bot^{k})$ along the path $e_{1} \cdots e_{j}$ in
  $\Tt_{G}$. If there are $j < j'$ such that $\struct{Q_{1}^{j},
    \ldots, Q_{n}^{j}} = \struct{Q_{1}^{j'}, \ldots, Q_{n}^{j'}}$,
  then removing the part of this sequence between $j$ and $j'$ will
  lead to the same final tuple $\struct{Q_{1}^{m}, \ldots,
    Q_{n}^{m}}$. We claim that after this removal, the resulting
  $k$-REM $e_{1} \cdots e_{j} \cdot e_{j'+1} \cdots e_{m}$ is a
  $k$-REM witness for $\struct{v_{p}, v_{q}}$ in $S$. The reason is as follows:
  from \lemref{lem:confGraphREMLang}, the following two conditions are
  equivalent to the original hypothesis that $e$ is a $k$-REM witness
  for $\struct{v_{p}, v_{q}}$ in $S$.
  \begin{enumerate}
  \item For some assignment $\s$, $(v_{q}, \s) \in Q_{p}^{m}$.
  \item For any $i = 1, \ldots, n$ and any $(v,\s) \in Q_{i}^{m}$,
    $(v_{i}, v) \in S$.
  \end{enumerate}
  Hence, any basic $k$-REM that ends in the same $n$-tuple
  $\struct{Q_{1}^{m}, ~\ldots~, Q_{n}^{m}}$ is also a $k$-REM witness
  for $\struct{v_{p}, v_{q}}$ in $S$.  As long as there are duplicate
  tuples along the sequence
  (\ref{eq:tuple-sequence-in-configuration-graph}), we can remove part
  of it to get a shorter witness. By pigeon hole principle, we
  conclude that there is a witness no longer than the total number of
  distinct tuples $\struct{Q_{1}, \ldots, Q_{n}}$.

  There are at most $n(\d+1)^{k}$ states in $\Tt_{G}$. Hence, there
  are at most $2^{n^{2}(\d+1)^{k}}$ tuples $\struct{Q_{1}, \ldots,
  Q_{n}}$. From the argument in the previous paragraph, we infer that
  if there is a $k$-REM witness for $\struct{v_{p}, v_{q}}$ in $S$,
  there is one of length at most $2^{n^{2}(\d+1)^{k}}$.
\end{proof}

For graphs without data considered in \cite{ANS2013}, the solution to
RPQ-definability looks at the graph as a finite automaton. This paves
the way for using language theoretic tools, which are
ultimately based on a pumping argument. In our case, we
cannot view a data graph directly as a register automaton. Hence we
need to construct the assignment graph on which we can apply the
pumping argument.

\begin{theorem}
    \label{thm:kREMDefinabilityBound}
    The $k-\rdpqm$-definability problem is in {\text{\sc
    NSpace}}$(\Oo(n^{2}\d^{k}))$, where $n$ is the number of nodes
    and $\d$ is the number of distinct data values.
\end{theorem}
\begin{proof}
    Suppose we are trying to define the set $S$. From
    \lemref{lem:union-of-basic}, it is enough to check that there are
    $|S|$ $k$-REM witness, one for each pair $\struct{u,v}$ in
    $S$. From \lemref{lem:small-witness}, we infer that it is enough
    to check for witnesses of length at most
    $2^{n^{2}(\d+1)^{k}}$. We will now give a non-deterministic
    algorithm to do this in space $\Oo(n^{2}\d^{k})$.

    First we note that given a tuple $\struct{Q_{1}, \ldots, Q_{n}}$
    of subsets of $Q_{G}$ and a $k$-REM $e=\downarrow
    \overline{r}a.[c]$, we can compute in space polynomial in
    $(n\d^{k})$ the tuple $\struct{Q_{1}', \ldots, Q_{n}'}$ such that
    $\struct{Q_{1}, \ldots, Q_{n}} \xra{~e~}_{G} \struct{Q_{1}',
    \ldots, Q_{n}'}$. Suppose $v_{1}, \ldots, v_{n}$ are the nodes of
    $G$. Now we give a non-deterministic algorithm to check if there
    exists a $k$-REM witness for $\struct{v_{1}, v_{p}}$ in $S$ of
    length at most $2^{n^{2}(\d+1)^{k}}$. The algorithm maintains a
    counter initialized to $0$ and a tuple of subsets of $Q_{G}$,
    initialized to $\struct{\{(v_{1}, \bot^{k})\}, \ldots,
    \{(v_{n}, \bot^{k})\}}$. The algorithm performs the following
    steps as long as the counter does not exceed
    $2^{n^{2}(\d+1)^{k}}$.
    \begin{enumerate}
        \item Increment the counter.
        \item Guess a $k$-REM $e = \downarrow\overline{r}a.[c]$.
        \item Replace current tuple $\struct{Q_{1}, \ldots,
            Q_{n}}$ with $\struct{Q_{1}', \ldots, Q_{n}'}$, where
            $\struct{Q_{1}, \ldots, Q_{n}} \xra{~e~} \struct{Q_{1}',
            \ldots, Q_{n}'}$.
        \item Check if $(v_{p}, \s) \in Q_{1}'$ for some assignment
            $\s$ and that for every $i = 1, \ldots, n$ and for every
            $(v',\s) \in Q_{i}'$, the pair $\struct{v_{i},v'}$ belongs
            to $S$. If yes, accept and
            terminate. If not, go back to step 1.
    \end{enumerate}

    From the proof of \lemref{lem:small-witness}, we conclude that
    some run of the above non-deterministic algorithm will accept if
    there is a $k$-REM witness for $\struct{v_{1}, v_{p}}$ in
    $S$. If there is no such witness, then clearly no run will accept.
    The algorithm needs space to store the counter, the tuple of
    subsets of $Q_{G}$ and the space to compute the successor tuple.
    The counter can be implemented in space $\Oo(n^{2}\d^{k})$ using
    binary counting. One state of $Q_{G}$ needs $(\log n \cdot k \cdot
    \log \d)$ bits. There are at most $n(\d+1)^{k}$ states in $Q_{G}$.
    Hence, the tuple of subsets and the space needed for intermediate
    computations can all be accommodated in space $\Oo(n^{2}\d^{k})$.
\end{proof}

\subsection{$\rdpqm$-definability}
\label{sec:rem-definability}

We can now tackle $\rdpqm$-definability, where there is no bound on
the number of registers.
\begin{lemma}
    \label{lem:definabilityMaxDRegisters}
    Suppose $G$ is a data graph with $\d$ distinct data values. A
    relation $S$ is $\rdpqm$-definable if and only if it is
    $\d$-$\rdpqm$-definable.
\end{lemma}
\begin{proof}
    The right to left implication is obvious. For the other direction,
    suppose $S$ is $k-\rdpqm$-definable for some $k$. From
    \lemref{lem:union-of-basic}, for every pair $\struct{u,v} \in S$,
    there is a $k$-REM witness $e$ for $\struct{u,v}$ in S. Hence,
    there exists a data path $w \in \Ll(e)$ connecting $u$ to
    $v$. From \lemref{lem:canonicalExpression}, $e_{[w]}$ is a
    $\d$-REM and from \factref{fact:autDataPaths}, $\Ll(e_{[w]}) \incl
    \Ll(e)$. Hence, $e_{[w]}$ is a $\d$-REM witness for
    $\struct{u,v}$ in $S$. Such witnesses exist for every pair in
    $S$ and hence, $S$ is $\d-\rdpqm$-definable.
\end{proof}

The next theorem follows from the previous two results.
\begin{theorem}
    \label{thm:REMDefinabilityExpspace}
    $\rdpqm$-definability is in \EXPSPACE{}.
\end{theorem}
\begin{proof}
    From \lemref{lem:definabilityMaxDRegisters}, it is equivalent to
    checking $\d-\rdpqm$-definability, where $\d$ is the number of
    distinct data values in the given data graph. From
    \thmref{thm:kREMDefinabilityBound}, $\d-\rdpqm$-definability is in
    {\text{\sc NSpace}}$(\Oo(n^{2}\d^{\d}))$. From Savitch's theorem,
    we then get a deterministic exponential space algorithm.
\end{proof}

Next we give a matching lower bound.

%%% Local Variables: 
%%% mode: latex
%%% TeX-master: "main"
%%% End: 

%\subsection{Lower bound}
%\label{sec:lower-bound-rem}
\begin{theorem}
    \label{thm:remLowBound}
    The $\rdpqm$-definability problem in data graphs is
    \EXPSPACE{}-hard.
\end{theorem}
\begin{proof}
    We reduce the exponential width corridor tiling problem to the
    $\rdpqm$-definability problem. An instance of the tiling problem consists
    of a set $T$ of tile types, a relation $C_{h} \subseteq T \times
    T$ of horizontally compatible tile types and a relation $C_{v}
    \subseteq T \times T$ of vertically compatible tile types, an
    initial tile type $t_{i}$, a final tile type $t_{f}$ and a number
    $n$ (in unary). The problem is to check if there exists a number
    $R$ and a tiling $\tau: [R]\times[2^{n}-1] \to T$ that is
    \emph{legal} --- $\tau(0,0) = t_{i}$, $\tau(R,2^{n}-1) = t_{f}$,
    $(\tau(i,j), \tau(i,j+1)) \in C_{h}$ and $(\tau(i,j), \tau(i+1,
    j)) \in C_{v}$ for all $i,j$. The intention here is that
    $\tau(i,j)$ is the tile type at the $i$\textsuperscript{th} row
    $j$\textsuperscript{th} column of a corridor with $R + 1$ rows and
    $2^{n}$ columns. This problem is known to be \EXPSPACE{}-complete
    (e.g., see \cite{Boas1997}). To be precise, we need to allow any
    exponential function in place of $2^{n}$. Our proof works in that
    case also; we use $2^{n}$ to reduce notational clutter.

    Let $\overline{T} = \{\overline{t} \mid t \in T\}$ be a disjoint
    copy of $T$. Given an instance of the tiling problem, we reduce it
    to the $\rdpqm$-definability problem in data graphs, where
    the finite alphabet is $T \cup \overline{T} \cup \{\$, \alpha\}$.
    A tiling $\tau$ is encoded by data paths in the language of the
    following REM:
    \begin{align}
        \begin{matrix}
            \$ \cdot & \downarrow r_{n} \cdot & \alpha \cdot &
            \downarrow r_{n-1} \cdot & \alpha \cdots \alpha \cdot &
            \downarrow r_{1} \cdot & \tau(0,0) & \\
            & [r_{n}^{=}] \cdot & \alpha & [r_{n-1}^{=}] \cdot &
            \alpha \cdots \alpha & [r_{1}^{\ne}] \cdot &
            \tau(0,1) & \\
            & [r_{n}^{=}] \cdot & \alpha & \cdots \alpha &
            [r_{2}^{\ne}] \cdot \alpha & [r_{1}^{=}] \cdot &
            \tau(0,2) & \\
            & & & & \vdots & & & \\
            & [r_{n}^{\ne}] \cdot & \alpha & [r_{n-1}^{\ne}] \cdot &
            \alpha \cdots \alpha & [r_{1}^{\ne}] \cdot &
            \overline{\tau(0,2^{n}-1)} & \\
            & [r_{n}^{=}] \cdot & \alpha & [r_{n-1}^{=}] \cdot &
            \alpha \cdots \alpha & [r_{1}^{=}] \cdot &
            \tau(1,0) & \\
            & & & & \vdots & & & \\
            & [r_{n}^{\ne}] \cdot & \alpha & [r_{n-1}^{\ne}] \cdot &
            \alpha \cdots \alpha & [r_{1}^{\ne}] \cdot &
            \overline{\tau(R, 2^{n}-1)} \cdot & \$
        \end{matrix}
        \label{eq:defRem}
    \end{align}
    The expression lists the tile types used in the tiling
    sequentially from left column to right column, bottom row to top
    row. The first $n$ data values are stored in the registers $r_{n},
    \ldots, r_{1}$. In later positions, $[r_{k}^{=}]$
    (resp.~$[r_{k}^{\ne}]$) indicates that the $k$\textsuperscript{th} bit
    is $0$ (resp.~$1$). The $n$ conditions preceding $\tau(i,j)$ in
    the expression denote the binary representation of $j$. Tile
    types in the last column are represented by letters in
    $\overline{T}$, so that we need not check them for horizontal
    compatibility with the next tile. The data graph is the
    disjoint union of two graphs $p_{1} \xrightarrow{\$}
    \fbox{\text{illegal tilings}} \xrightarrow{\$} q_{1}$ and $p_{2}
    \xrightarrow{\$} \fbox{\text{all tilings}} \xrightarrow{\$} q_{2}$
    satisfying the following conditions.
    \begin{enumerate}
        \item \label{it:connectivityGuard} Any data path starting and
            ending with the letter $\$$ may only connect $p_{1}$ to
            $q_{1}$ or $p_{2}$ to $q_{2}$.
        \item \label{it:allTilingsEncoded} Every tiling can be encoded
            by some data path connecting $p_{2}$ to $q_{2}$.
        \item \label{it:noLegalTilingsEncoded} None of the data paths
            connecting $p_{1}$ to $q_{1}$ are encodings of legal tilings.
        \item \label{it:allIllegalsEncoded} For every data path $w$
            connecting $p_{2}$ to $q_{2}$ that is not the encoding of
            a legal tiling, there exists a data path automorphic to
            $w$ connecting $p_{1}$ to $q_{1}$.
    \end{enumerate}

    We claim that there exists a legal tiling iff $\{\struct{p_{2},
    q_{2}}\}$
    is $\rdpqm$-definable. Indeed, suppose there exists a legal tiling
    $\tau$. Conditions \ref{it:connectivityGuard},
    \ref{it:allTilingsEncoded} and \ref{it:noLegalTilingsEncoded}
    ensure that the REM in \eqref{eq:defRem} defines $\{\struct{p_{2},
    q_{2}}\}$. Conversely, suppose $\{\struct{p_{2}, q_{2}}\}$ is definable.
    There exists a defining REM $e$ and a data path $w$ in
    $\lang(e)$ connecting $p_{2}$ to $q_{2}$. If $w$ does not encode
    a legal tiling, then condition \ref{it:allIllegalsEncoded} above
    implies that there is a data path $w'$ automorphic
    to $w$ (and hence in $\lang(e)$) connecting $p_{1}$
    to $q_{1}$, contradicting the hypothesis that $e$ defines
    $\{\struct{p_{2}, q_{2}}\}$. Hence, $w$ encodes a legal tiling.
    The data graph can be constructed in polynomial time (details
    follow) and hence the $\rdpqm$-definability problem in data graphs
    is \EXPSPACE{}-hard. At a high level, the strategy of this proof is similar to
    that of \cite[Theorem 3.7]{KRV2014} in the sense that a small
    gadget differentiates between the set of all tilings and the set
    of illegal tilings. However, \cite[Theorem 3.7]{KRV2014} can not
    be used here directly, since that is about containment of one
    query in another while we are concerned about the definability of
    a relation in a given data graph. There is also a subtle
    difference between the proof strategies which will be
    highlighted in the details that follow.

    We now give the details of the data graph. Nodes are denoted
    by circles, with data values written outside. The data values of
    some nodes are skipped when they are not important. The data
    values $d_{n}, e_{n}, \ldots, d_{1}, e_{1}$ are all distinct. The
    portion of the data graph containing $p_{2}$ and $q_{2}$ is as
    follows.

    
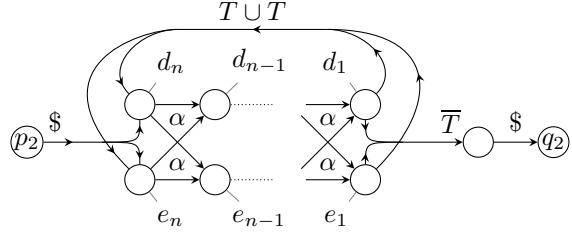
\begin{center}
        \begin{tikzpicture}[>=stealth, pin distance=0.1\ml]
    \node[state] (p2) at (0\ml, 0\ml) {$p_{2}$};
    \node[state, pin=70:$d_{n}$] (dn) at ([xshift=1.5\ml,yshift=0.5\ml]p2) {};
    \node[state, pin=70:$d_{n-1}$] (dnm1) at ([xshift=1\ml]dn) {};
    \node[state, pin=110:$d_{1}$] (d1) at ([xshift=2\ml]dnm1) {};
    \node[state, draw=none] (d2) at ([xshift=-1\ml]d1) {};
    \node[state, pin=-70:$e_{n}$] (en) at ([xshift=1.5\ml,yshift=-0.5\ml]p2) {};
    \node[state, pin=-70:$e_{n-1}$] (enm1) at (en -| dnm1) {};
    \node[state, pin=-110:$e_{1}$] (e1) at (en -| d1) {};
    \node[state, draw=none] (e2) at ([xshift=-1\ml]e1) {};
    \node[state] (lt) at ([xshift=1.5\ml]barycentric cs:d1=1,e1=1) {};
    \node[state] (q2) at ([xshift=1\ml]lt) {$q_{2}$};
    \node[coordinate] (ft) at ([xshift=1\ml]p2) {};

    \begin{scope}[decoration={markings,mark = at position 0.5 with {\arrow{stealth}}}]
        \draw[postaction={decorate}] (p2) -- node[auto=left, pos=0.2] {$\$$} (ft);
    \end{scope}
    \begin{scope}[decoration={markings,mark = at position 0.9 with {\arrow{stealth}}}]
        \draw [postaction={decorate}](ft) .. controls ([xshift=0.5\ml]ft) and ([yshift=-0.5\ml]dn.center) .. (dn);
        \draw [postaction={decorate}](ft) .. controls ([xshift=0.5\ml]ft) and ([yshift=0.5\ml]en.center) .. (en);
    \end{scope}
    \draw[->] (dn) -- node[auto=right] {$\alpha$} (dnm1);
    \draw[->] (en) -- node[auto=left] {$\alpha$} (enm1);
    \draw[->] (dn) -- (enm1);
    \draw[->] (en) -- (dnm1);
    \draw[->] (d2) -- node[auto=right] {$\alpha$} (d1);
    \draw[->] (e2) -- node[auto=left] {$\alpha$} (e1);
    \draw[->] (d2) -- (e1);
    \draw[->] (e2) -- (d1);
    \draw[densely dotted] (dnm1) -- (d2);
    \draw[densely dotted] (enm1) -- (e2);

    \node[coordinate] (t1) at ([xshift=-0.5\ml,yshift=1\ml]d1) {};
    \node[coordinate] (t2) at ([xshift=0.5\ml,yshift=1\ml]dn) {};
    \begin{scope}[decoration={markings,mark = at position 0.5 with {\arrow{stealth}}}]
        \draw [postaction={decorate}](t1) -- node[auto=right] {$T \cup \overline{T}$} (t2);
        \draw [postaction={decorate}](d1) .. controls ([xshift=0.5\ml, yshift=0.5\ml]d1.center) and ([xshift=0.5\ml]t1) .. (t1);
        \draw [postaction={decorate}](e1) .. controls ([xshift=1\ml, yshift=1\ml]e1.center) and ([xshift=1.5\ml]t1) .. (t1);
    \end{scope}
    \begin{scope}[decoration={markings,mark = at position 0.9 with {\arrow{stealth}}}]
        \draw [postaction={decorate}](t2) .. controls ([xshift=-0.5\ml]t2) and ([xshift=-0.5\ml, yshift=0.5\ml]dn.center) .. (dn);
        \draw [postaction={decorate}](t2) .. controls ([xshift=-0.5\ml]t2) and ([xshift=-1.5\ml, yshift=1.5\ml]en.center) .. (en);
    \end{scope}

    \node[coordinate] (t3) at ([xshift=-1\ml]lt) {};
    \begin{scope}[decoration={markings,mark = at position 0.2 with {\arrow{stealth}}}]
        \draw [postaction={decorate}](d1) .. controls ([yshift=-0.5\ml]d1.center) and ([xshift=-0.5\ml]t3) .. (t3);
        \draw [postaction={decorate}](e1) .. controls ([yshift=0.5\ml]e1.center) and ([xshift=-0.5\ml]t3) .. (t3);
    \end{scope}
    \draw[->] (t3) -- node[auto=left, pos=0.8] {$\overline{T}$} (lt);
    \draw[->] (lt) -- node[auto=left] {$\$$} (q2);
\end{tikzpicture}
    \end{center}

    For every data path connecting $p_{2}$ to $q_{2}$ that is not the
    encoding of a legal tiling, we now add automorphic data paths
    connecting $p_{1}$ to $q_{1}$ through gadgets. This will ensure
    that the data graph satisfies condition
    \ref{it:allIllegalsEncoded} above. We also ensure that every path
    we add satisfies conditions \ref{it:connectivityGuard} and
    \ref{it:noLegalTilingsEncoded}.    

    \textbullet\hphantom{a} \emph{In a data path $w$ connecting $p_{2}$ to
            $q_{2}$, the sequence of $n$ data values preceding
            $\tau(0,1)$ does not represent $1$}.  This could be due to any one
            (or more) of the $n$ bits being wrong; following is the gadget
            for checking that the $k$\textsuperscript{th} bit is $1$ (at
            node $q$) instead of $0$. There are $n$ such gadgets, one for
            each bit.
            
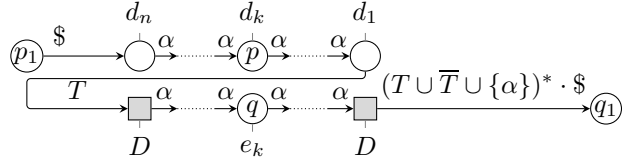
\begin{center}
                \begin{tikzpicture}[>=stealth, pin distance=0.1\ml]
    \node[state] (p1) at (0\ml,0\ml) {$p_{1}$};
    \node[state, pin=90:$d_{n}$] (dn) at ([xshift=1.5\ml]p1) {};
    \node[state, pin=90:$d_{k}$] (di) at ([xshift=1.5\ml]dn) {$p$};
    \node[state, pin=90:$d_{1}$] (d1) at ([xshift=1.5\ml]di) {};
    \node[gadget, pin=-90:$D$, fill=gray!30] (dnp) at ([yshift=-0.7\ml]dn) {};
    \node[state, pin=-90:$e_{k}$] (dip) at ([xshift=1.5\ml]dnp) {$q$};
    \node[gadget, pin=-90:$D$, fill=gray!30] (d1p) at ([xshift=1.5\ml]dip) {};
    \node[state] (q1) at ([xshift=3.2\ml]d1p) {$q_{1}$};

    \draw[->] (p1) -- node[auto=left, pos=0.2] {$\$$} (dn);
    \draw[->] (dn) -- node[auto=left] {$\alpha$} ([xshift=0.5\ml]dn.center);
    \draw[densely dotted] ([xshift=0.5\ml]dn.center) -- ([xshift=-0.5\ml]di.center);
    \draw[->] ([xshift=-0.5\ml]di.center) -- node[auto=left] {$\alpha$} (di);
    \draw[->] (di) -- node[auto=left] {$\alpha$} ([xshift=0.5\ml]di.center);
    \draw[densely dotted] ([xshift=0.5\ml]di.center) -- ([xshift=-0.5\ml]d1.center);
    \draw[->] ([xshift=-0.5\ml]d1.center) -- node[auto=left] {$\alpha$} (d1);
    \draw[->, rounded corners=0.05\ml] (d1) -- ([yshift=-0.3\ml]d1.center) -- ([yshift=-0.3\ml]p1.center) -- (p1 |- dnp) -- node[auto=left] {$T$} (dnp);
    \draw[->] (dnp) -- node[auto=left] {$\alpha$} ([xshift=0.5\ml]dnp.center);
    \draw[densely dotted] ([xshift=0.5\ml]dnp.center) -- ([xshift=-0.5\ml]dip.center);
    \draw[->] ([xshift=-0.5\ml]dip.center) -- node[auto=left] {$\alpha$} (dip);
    \draw[->] (dip) -- node[auto=left] {$\alpha$} ([xshift=0.5\ml]dip.center);
    \draw[densely dotted] ([xshift=0.5\ml]dip.center) -- ([xshift=-0.5\ml]d1p.center);
    \draw[->] ([xshift=-0.5\ml]d1p.center) -- node[auto=left] {$\alpha$} (d1p);
    \draw[->] (d1p) -- node[auto=left] {$(T \cup \overline{T} \cup \{\a\})^{*} \cdot \$$} (q1);
\end{tikzpicture}\\*
            \end{center}
            The gadget simply checks that the $k$\textsuperscript{th}
            data value preceding $\tau(0,1)$ in the data path $w$
            ($e_{k}$ at node $q$) is unequal to the $k$\textsuperscript{th}
            data value preceding $\tau(0,0)$ ($d_{k}$ at node $p$).
            This will ensure that the $k$\textsuperscript{th} bit
            preceding $\tau(0,1)$ is $1$ instead of $0$.  If the
            $k$\textsuperscript{th} data value preceding $\tau(0,0)$
            is $e_{k}$, the gadget can still imitate $w$ modulo an
            automorphism that interchanges $d_{k}$ and $e_{k}$.  In
            the above diagram, every gray box marked $D$ is actually a
            gadget with $2n$ nodes, with each node having a distinct
            data value from the set $\{d_{n}, e_{n}, \ldots, d_{1},
            e_{1}\}$. For every edge coming in to a gray box, there is
            an edge coming in to each of the $2n$ nodes. For every
            edge coming out of a gray box, there is an edge coming out
            of each of the $2n$ nodes. The edge coming in to the node
            $q_{1}$ is labeled by the REM $(T \cup \overline{T} \cup
            \{\a\})^{*} \cdot \$$, whose language consists of all the
            data paths having exactly one occurrence of the letter
            $\$$, which occurs at the end. A gadget admitting exactly
            this set of data paths can be easily designed using
            polynomially many nodes; the gadget is not shown in the
            diagram since it is easier to understand the expression.

        \textbullet\hphantom{a} \emph{Some sequence of $n$ conditions does not encode
            the successor of the preceding $n$ conditions}. The
            following gadget checks that the $k$\textsuperscript{th}
            bit flips from $1$ (in node $2$) to $0$ (in node $4$) but
            the $(k+1)$\textsuperscript{th} bit stays at $0$ (in nodes
            $1$ and $3$). There are $\Oh(n)$ such gadgets for checking
            all such errors.
            
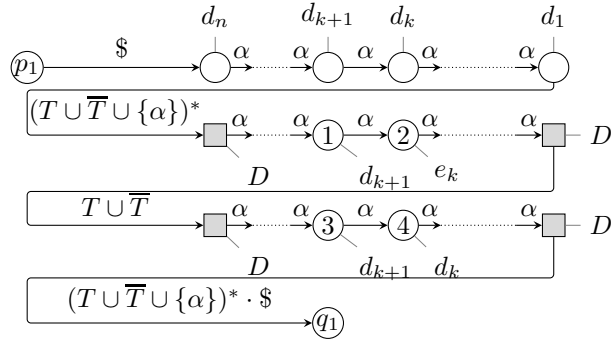
\begin{center}
                \begin{tikzpicture}[>=stealth, pin distance=0.2\ml]
    \node[state] at (0\ml,0\ml) {$p_{1}$};
    \node[state,pin=90:$d_{n}$] (dn1) at ([xshift=2.5\ml]p1) {};
    \node[state,pin=90:$d_{k+1}$] (dip11) at ([xshift=2.5\ml]dn) {};
    \node[state,pin=90:$d_{k}$] (di1) at ([xshift=1\ml]dip11) {};
    \node[state,pin=90:$d_{1}$] (d11) at ([xshift=2\ml]di1) {};
    \node[gadget,pin=-45:$D$, fill=gray!30] (dn2) at ([yshift=-0.9\ml]dn1) {};
    \node[state,pin=-45:$d_{k+1}$] (dip12) at (dip11 |- dn2) {$1$};
    \node[state,pin=-45:$e_{k}$] (di2) at (di1 |- dn2) {$2$};
    \node[gadget,pin=0:$D$, fill=gray!30] (d12) at (d11 |- dn2) {};
    \node[gadget,pin=-45:$D$, fill=gray!30] (dn3) at ([yshift=-1.2\ml]dn2) {};
    \node[state,pin=-45:$d_{k+1}$] (dip13) at (dip12 |- dn3) {$3$};
    \node[state,pin=-45:$d_{k}$] (di3) at (di2 |- dn3) {$4$};
    \node[gadget,pin=0:$D$, fill=gray!30] (d13) at (d12 |- dn3) {};
    \node[state] (q1) at ([yshift=-1.3\ml]dip13) {$q_{1}$};

    \draw[->] (p1) -- node[auto=left] {$\$$} (dn1);
    \draw[->] (dn1) -- node[auto=left] {$\alpha$} ([xshift=0.5\ml]dn1.center);
    \draw[densely dotted] ([xshift=0.5\ml]dn1.center) -- ([xshift=-0.5\ml]dip11.center);
    \draw[->] ([xshift=-0.5\ml]dip11.center) -- node[auto=left] {$\alpha$} (dip11);
    \draw[->] (dip11) -- node[auto=left] {$\alpha$} (di1);
    \draw[->] (di1) -- node[auto=left] {$\alpha$} ([xshift=0.5\ml]di1.center);
    \draw[densely dotted] ([xshift=0.5\ml]di1.center) -- ([xshift=-0.5\ml]d11.center);
    \draw[->] ([xshift=-0.5\ml]d11.center) -- node[auto=left] {$\alpha$} (d11);
    \draw[->, rounded corners=0.05\ml] (d11) -- ([yshift=-0.3\ml]d11.center) -- ([yshift=-0.3\ml]p1.center) -- (p1|- dn2) -- node[auto=left] {$(T \cup \overline{T} \cup \{\a\})^{*}$} (dn2);
    \draw[->] (dn2) -- node[auto=left] {$\alpha$} ([xshift=0.5\ml]dn2.center);
    \draw[densely dotted] ([xshift=0.5\ml]dn2.center) -- ([xshift=-0.5\ml]dip12.center);
    \draw[->] ([xshift=-0.5\ml]dip12.center) -- node[auto=left] {$\alpha$} (dip12);
    \draw[->] (dip12) -- node[auto=left] {$\alpha$} (di2);
    \draw[->] (di2) -- node[auto=left] {$\alpha$} ([xshift=0.5\ml]di2.center);
    \draw[densely dotted] ([xshift=0.5\ml]di2.center) -- ([xshift=-0.5\ml]d12.center);
    \draw[->] ([xshift=-0.5\ml]d12.center) -- node[auto=left] {$\alpha$} (d12);
    \draw[->, rounded corners=0.05\ml] (d12) -- ([yshift=-0.75\ml]d12.center) -- ([yshift=-0.75\ml]p1 |- dn2) -- (p1 |- dn3) -- node[auto=left] {$T \cup \overline{T}$} (dn3);
    \draw[->] (dn3) -- node[auto=left] {$\alpha$} ([xshift=0.5\ml]dn3.center);
    \draw[densely dotted] ([xshift=0.5\ml]dn3.center) -- ([xshift=-0.5\ml]dip13.center);
    \draw[->] ([xshift=-0.5\ml]dip13.center) -- node[auto=left] {$\alpha$} (dip13);
    \draw[->] (dip13) -- node[auto=left] {$\alpha$} (di3);
    \draw[->] (di3) -- node[auto=left] {$\alpha$} ([xshift=0.5\ml]di3.center);
    \draw[densely dotted] ([xshift=0.5\ml]di3.center) -- ([xshift=-0.5\ml]d13.center);
    \draw[->] ([xshift=-0.5\ml]d13.center) -- node[auto=left] {$\alpha$} (d13);
    \draw[->, rounded corners=0.05\ml] (d13) -- ([yshift=-0.7\ml]d13.center) -- ([yshift=-0.7\ml]d13 -| p1) -- (p1 |- q1) -- node[auto=left] {$(T \cup \overline{T} \cup \{\a\})^{*} \cdot \$$} (q1);
\end{tikzpicture}
            \end{center}

        \textbullet\hphantom{a} \emph{The sequence of $n$ conditions before a letter in
                $\overline{T}$ does not represent $2^{n}-1$}. This can
                be due to any one (or more) of the bits being $0$
                instead of $1$. The following gadget checks that the
                $k$\textsuperscript{th} bit is $0$ (at node $q$).
                There are $n$ such gadgets, one for each bit.\\*
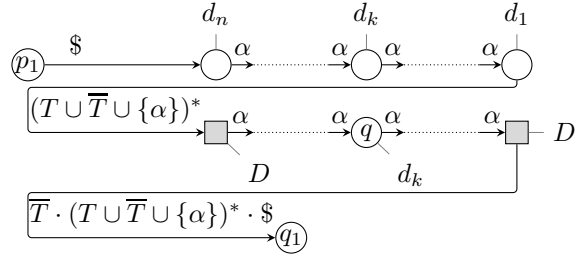
                \begin{center}
                    \begin{tikzpicture}[>=stealth, pin distance=0.2\ml]
    \node[state] at (0\ml,0\ml) {$p_{1}$};
    \node[state,pin=90:$d_{n}$] (dn1) at ([xshift=2.5\ml]p1) {};
    \node[state,pin=90:$d_{k}$] (di1) at ([xshift=2\ml]dn1) {};
    \node[state,pin=90:$d_{1}$] (d11) at ([xshift=2\ml]di1) {};
    \node[gadget,pin=-45:$D$, fill=gray!30] (dn2) at ([yshift=-0.9\ml]dn1) {};
    \node[state,pin=-45:$d_{k}$] (di2) at (di1 |- dn2) {$q$};
    \node[gadget,pin=0:$D$, fill=gray!30] (d12) at (d11 |- dn2) {};
    \node[state] (q1) at ([xshift=1\ml, yshift=-1.4\ml]dn2) {$q_{1}$};
    
    \draw[->] (p1) -- node[auto=left, pos=0.2] {$\$$} (dn1);
    \draw[->] (dn1) -- node[auto=left] {$\alpha$} ([xshift=0.5\ml]dn1.center);
    \draw[densely dotted] ([xshift=0.5\ml]dn1.center) -- ([xshift=-0.5\ml]di1.center);
    \draw[->] ([xshift=-0.5\ml]di1.center) -- node[auto=left] {$\alpha$} (di1);
    \draw[->] (di1) -- node[auto=left] {$\alpha$} ([xshift=0.5\ml]di1.center);
    \draw[densely dotted] ([xshift=0.5\ml]di1.center) -- ([xshift=-0.5\ml]d11.center);
    \draw[->] ([xshift=-0.5\ml]d11.center) -- node[auto=left] {$\alpha$} (d11);
    \draw[->, rounded corners=0.05\ml] (d11) -- ([yshift=-0.3\ml]d11.center) -- ([yshift=-0.3\ml]p1.center) -- (p1 |- dn2) -- node[auto=left] {$(T \cup \overline{T} \cup \{\a\})^{*}$} (dn2);
    \draw[->] (dn2) -- node[auto=left] {$\alpha$} ([xshift=0.5\ml]dn2.center);
    \draw[densely dotted] ([xshift=0.5\ml]dn2.center) -- ([xshift=-0.5\ml]di2.center);
    \draw[->] ([xshift=-0.5\ml]di2.center) -- node[auto=left] {$\alpha$} (di2);
    \draw[->] (di2) -- node[auto=left] {$\alpha$} ([xshift=0.5\ml]di2.center);
    \draw[densely dotted] ([xshift=0.5\ml]di2.center) -- ([xshift=-0.5\ml]d12.center);
    \draw[->] ([xshift=-0.5\ml]d12.center) -- node[auto=left] {$\alpha$} (d12);
    \draw[->, rounded corners=0.05\ml] (d12) -- ([yshift=-0.8\ml]d12.center) -- ([yshift=-0.8\ml]d12 -|p1) -- (p1 |- q1) -- node[auto=left] {$\overline{T} \cdot (T \cup \overline{T} \cup \{\a\})^{*} \cdot \$$} (q1);
\end{tikzpicture}
                \end{center}

        \textbullet\hphantom{a} \emph{The $n$ conditions before a letter in
                $T$ represent $2^{n} - 1$}.
                
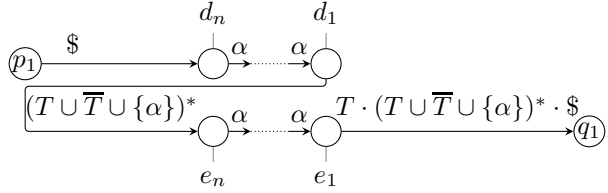
\begin{center}
                    \begin{tikzpicture}[>=stealth, pin distance = 0.2\ml]
    \node[state] at (0\ml,0\ml) {$p_{1}$};
    \node[state,pin=90:$d_{n}$] (dn1) at ([xshift=2.5\ml]p1) {};
    \node[state,pin=90:$d_{1}$] (d11) at ([xshift=1.5\ml]dn1) {};
    \node[state,pin=-90:$e_{n}$] (dn2) at ([yshift=-0.9\ml]dn1) {};
    \node[state,pin=-90:$e_{1}$] (d12) at (d11 |- dn2) {};
    \node[state] (q1) at ([xshift=3.5\ml]d12) {$q_{1}$};
    
    \draw[->] (p1) -- node[auto=left, pos=0.2] {$\$$} (dn1);
    \draw[->] (dn1) -- node[auto=left] {$\alpha$} ([xshift=0.5\ml]dn1.center);
    \draw[densely dotted] ([xshift=0.5\ml]dn1.center) -- ([xshift=-0.5\ml]d11.center);
    \draw[->] ([xshift=-0.5\ml]d11.center) -- node[auto=left] {$\alpha$} (d11);
    \draw[->, rounded corners=0.05\ml] (d11) -- ([yshift=-0.3\ml]d11.center) -- ([yshift=-0.3\ml]p1.center) -- (p1 |- dn2) -- node[auto=left] {$(T \cup \overline{T} \cup \{\a\})^{*}$} (dn2);
    \draw[->] (dn2) -- node[auto=left] {$\alpha$} ([xshift=0.5\ml]dn2.center);
    \draw[densely dotted] ([xshift=0.5\ml]dn2.center) -- ([xshift=-0.5\ml]d12.center);
    \draw[->] ([xshift=-0.5\ml]d12.center) -- node[auto=left] {$\alpha$} (d12);
    \draw[->] (d12) -- node[auto=left] {$T \cdot (T \cup \overline{T} \cup \{\a\})^{*} \cdot \$$} (q1);
\end{tikzpicture}
                \end{center}

        \textbullet\hphantom{a} \emph{The tiling does not begin with the tile type
                $t_{i}$}.

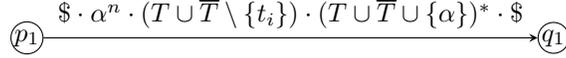
                \begin{center}
                    \begin{tikzpicture}[>=stealth, pin distance=0.2\ml]
    \node[state] (p1) at (0\ml,0\ml) {$p_{1}$};
    \node[state] (q1) at ([xshift=7\ml]p1) {$q_{1}$};

    \draw[->] (p1) -- node[auto=left] {$\$ \cdot \a^{n} \cdot (T \cup \overline{T} \setminus \{t_{i}\}) \cdot (T \cup \overline{T} \cup \{\a\})^{*} \cdot \$$} (q1);
\end{tikzpicture}

                \end{center}

        \textbullet\hphantom{a} \emph{The tiling does not end with the tile type
                $t_{f}$}.

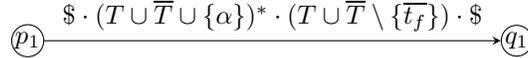
                \begin{center}
                    \begin{tikzpicture}[>=stealth, pin distance=0.2\ml]
    \node[state] (p1) at (0\ml,0\ml) {$p_{1}$};
    \node[state] (q1) at ([xshift=6.5\ml]p1) {$q_{1}$};

    \draw[->] (p1) -- node[auto=left] {$\$ \cdot (T \cup \overline{T} \cup \{\a\})^{*} \cdot (T \cup \overline{T} \setminus \{\overline{t_{f}}\}) \cdot \$$} (q1);
\end{tikzpicture}

                \end{center}

        \textbullet\hphantom{a} \emph{Two adjacent tiles in the same row are not
            horizontally compatible}. The following gadget checks that
            the tile type $t_{2}$ is adjacent to the horizontally
            incompatible tile type $t_{1}$ in the same row. There is
            one such gadget for every pair $(t_{1}, t_{2})$ of
            horizontally incompatible tile types.

            
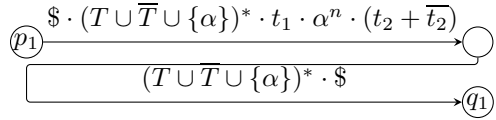
\begin{center}
                \begin{tikzpicture}[>=stealth, pin distance=0.2\ml]
    \node[state] (p1) at (0\ml,0\ml) {$p_{1}$};
    \node[state] (p11) at ([xshift=6\ml]p1) {};
    \node[state] (q1) at ([xshift=6\ml, yshift=-0.8\ml]p1) {$q_{1}$};

    \draw[->] (p1) --node[auto=left] {$\$ \cdot (T \cup \overline{T}\cup \{\a\})^{*} \cdot t_{1} \cdot \a^{n} \cdot
    (t_{2} + \overline{t_{2}}) $} (p11);
    \draw[->, rounded corners=0.05\ml] (p11)-- ([yshift=-0.3\ml]p11.center) -- ([yshift=-0.3\ml]p11 -| p1) -- (p1 |- q1) -- node[auto=left] {$(T \cup \overline{T} \cup \{\a\})^{*} \cdot \$$}
    (q1);
\end{tikzpicture}
            \end{center}

        \textbullet\hphantom{a} \emph{Two adjacent tiles in the same column are not
            vertically compatible}. The gadget below checks that the
            tile type $t_{2}$ (seen just after node $2$) is adjacent
            to the vertically incompatible tile type $t_{1}$(seen just
            after node $1$) in the last column. The tiles are matched
            from the same column, since the data values seen just
            before $\overline{t_{1}}$ are same as the data values seen
            just before $\overline{t_{2}}$. The tiles are matched from
            adjacent rows, since $\overline{t_{1}}$ is the only letter
            from $\overline{T}$ allowed between the nodes $1$ and $2$.
            There is one such gadget for every pair $(t_{1}, t_{2})$
            of vertically incompatible tile types.
            
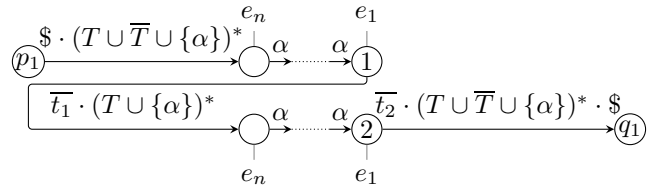
\begin{center}
                \begin{tikzpicture}[>=stealth, pin distance=0.2\ml]
    \node[state] (p1) at (0\ml,0\ml) {$p_{1}$};
    \node[state, pin=90:$e_{n}$] (dn1) at ([xshift=3\ml]p1) {};
    \node[state, pin=90:$e_{1}$] (d11) at ([xshift=1.5\ml]dn1) {$1$};
    \node[state, pin=-90:$e_{n}$] (dn2) at ([yshift=-0.9\ml]dn1) {};
    \node[state, pin=-90:$e_{1}$] (d12) at (d11 |- dn2) {$2$};
    \node[state] (q1) at ([xshift=3.5\ml]d12) {$q_{1}$};

    \draw[->] (p1) -- node[auto=left] {$\$ \cdot (T \cup \overline{T} \cup \{\a\})^{*}$} (dn1);
    \draw[->] (dn1) -- node[auto=left] {$\alpha$} ([xshift=0.5\ml]dn1.center);
    \draw[densely dotted] ([xshift=0.5\ml]dn1.center) -- ([xshift=-0.5\ml]d11.center);
    \draw[->] ([xshift=-0.5\ml]d11.center) -- node[auto=left] {$\alpha$} (d11);
    \draw[->, rounded corners=0.05\ml] (d11) -- ([yshift=-0.3\ml]d11.center) -- ([yshift=-0.3\ml]p1.center) -- (p1 |- dn2) -- node[auto=left] {$\overline{t_{1}} \cdot (T \cup \{\a\})^{*}$} (dn2);
    \draw[->] (dn2) -- node[auto=left] {$\alpha$} ([xshift=0.5\ml]dn2.center);
    \draw[densely dotted] ([xshift=0.5\ml]dn2.center) -- ([xshift=-0.5\ml]d12.center);
    \draw[->] ([xshift=-0.5\ml]d12.center) -- node[auto=left] {$\alpha$} (d12);
    \draw[->] (d12) -- node[auto=left] {$\overline{t_{2}} \cdot (T \cup \overline{T} \cup \{\a\})^{*} \cdot \$$} (q1);
\end{tikzpicture}
            \end{center}
            We now highlight a subtle difference between this proof
            and that of \cite[Theorem 3.7]{KRV2014}. There, to check
            that the distance between two positions is exactly
            $2^{n}$, a gadget checks $n$ bits of the first position
            against $n$ bits of the second position. There the gadget
            is built using REM and hence it can check the bits
            individually for equality. Here, we need to build a
            similar gadget using data graphs.  If we did this check by
            matching each bit explicitly, the gadget would be
            exponentially larger. We avoid it by observing that we
            need not admit the exact data path encoding an illegal
            tiling, but only an \emph{automorphic copy}. Thus, if a
            data path has the data value $d_{1}$ just before
            $\overline{t_{1}}$ and $\overline{t_{2}}$, the above
            gadget will still catch it through an automorphism that
            interchanges $d_{1}$ and $e_{1}$.

            The gadget below checks that the tile type $t_{2}$
            is adjacent to the vertically incompatible tile type
            $t_{1}$ in a column other than the last one. There
            is one such gadget for every pair $(t_{1}, t_{2})$ of
            vertically incompatible tile types.

            
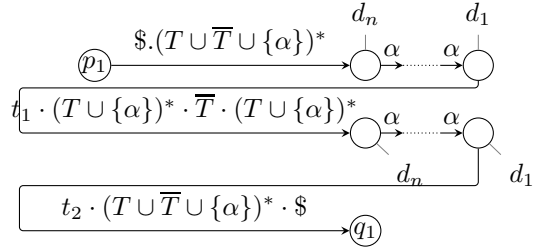
\begin{center}
                \begin{tikzpicture}[>=stealth, pin distance=0.2\ml]
    \node[state] (p1) at (0\ml,0\ml) {$p_{1}$};
    \node[state, pin=90:$d_{n}$] (dn1) at ([xshift=3.6\ml]p1) {};
    \node[state, pin=90:$d_{1}$] (d11) at ([xshift=1.5\ml]dn1) {};
    \node[state, pin=-45:$d_{n}$] (dn2) at ([yshift=-0.9\ml]dn1) {};
    \node[state, pin=-45:$d_{1}$] (d12) at (d11 |- dn2) {};
    \node[state] (q1) at ([yshift=-1.3\ml]dn2) {$q_{1}$};

    \draw[->] (p1) -- node[auto=left] {$\$. (T \cup \overline{T} \cup \{\a\})^{*}$} (dn1);
    \draw[->] (dn1) -- node[auto=left] {$\alpha$} ([xshift=0.5\ml]dn1.center);
    \draw[densely dotted] ([xshift=0.5\ml]dn1.center) -- ([xshift=-0.5\ml]d11.center);
    \draw[->] ([xshift=-0.5\ml]d11.center) -- node[auto=left] {$\alpha$} (d11);
    \draw[->, rounded corners = 0.05\ml] (d11) -- ([yshift=-0.3\ml]d11.center) -- ([xshift=-1\ml,yshift=-0.3\ml]p1.center)--
    ([xshift=-1\ml]p1 |- dn2) -- node[auto=left] {$t_{1} \cdot (T \cup \{\a\})^{*} \cdot \overline{T} \cdot (T \cup \{\a\})^{*}$} (dn2);
    \draw[->] (dn2) -- node[auto=left] {$\alpha$} ([xshift=0.5\ml]dn2.center);
    \draw[densely dotted] ([xshift=0.5\ml]dn2.center) -- ([xshift=-0.5\ml]d12.center);
    \draw[->] ([xshift=-0.5\ml]d12.center) -- node[auto=left] {$\alpha$} (d12);
    \draw[->, rounded corners=0.05\ml] (d12) -- ([yshift=-0.7\ml]d12.center) -- ([xshift=-1\ml,yshift=-0.7\ml]d12 -| p1) -- ([xshift=-1\ml]p1 |- q1) -- node[auto=left] {$t_{2} \cdot (T \cup \overline{T} \cup \{\a\})^{*} \cdot \$$} (q1);
\end{tikzpicture}
            \end{center}
\end{proof}

%%% Local Variables: 
%%% mode: latex
%%% TeX-master: "main"
%%% End: 

\newcommand{\bin}{\textsf{Bin}}
\newcommand{\defn}{\textsf{Def}}
\newcommand{\comp}{\circ}
% REE
\section{Queries using Regular expressions with equality}

We study the $\rdpqe$-definability problem in this section. These are
queries using REE. As seen before, REE have the additional $e_{=}$ and
$e_{\neq}$ constructs on top of standard regular expressions. In
Example~\ref{eg:queries-REM-REE}, we have seen that they are less
powerful than REM in defining relations. However, due to $e_{=}$ and
$e_{\neq}$, they can define more relations than RPQs.

In Section 3, we have shown that $\rdpqm$-definability is
$\EXPSPACE$-complete.
RPQ-definability is known to be
$\PSPACE$-complete~\cite{ANS2013}. We will now prove that
$\rdpqe$-definability is $\PSPACE$-complete as well.

The main idea is the following observation. Suppose we have an
expression $e_1 \cdot e_2$. If $e_1$ and $e_2$ are REMs, it is possible
that there is a register getting bound in $e_1$ by $\downarrow r$, and
used in a condition in $e_2$. So one cannot reason about $e_1\cdot
e_2$ by independently reasoning about $e_1$ and $e_2$. Such a
situation does not arise with REE. The relation defined by $e_1 \cdot
e_2$ can in fact be obtained as a composition of the relations defined
by $e_1$ and $e_2$. This makes it possible to solve the problem in \PSPACE{}.

For this section, fix a data graph $G = (V, E, \rho)$ over edge
alphabet $\Sigma$.
Let $\bin = V \times V$ be the set of binary relations over $V$. Our
first goal would be to define operators over the set $\bin$ and to
generate relations in a hierarchical manner by making use of
these operators in a certain way.

%%% \todo{restriction instead of projection?}

% Some of the relations in $\bin$ are REE-definable and some of them
% are not. Our first goal would be to divide the REE-definable
% relations into different levels based on the number of nested
% equality checks needed for an REE to define it.

\begin{definition}
  Given two relations $S_1, S_2 \in \bin$, we define the following
  operators:
  \begin{align*}
    S_1 + S_2 & ~=~ \{ \struct{u,v}~|~ \struct{u,v} \in S_1 \text{ or
    } \struct{u,v} \in
    S_2 \}     \\
    S_1 \comp S_2 & ~=~ \{ \struct{u,v}~|~ \exists z: \struct{u,z} \in
    S_1 \text{
      and } \struct{z,v} \in S_2 \} \\
    S^= & ~=~ \{ \struct{u,v} \in S~|~ \rho(u) = \rho(v) \} \\
    S^{\neq} & ~=~ \{ \struct{u,v} \in S~|~ \rho(u) \neq \rho(v) \}
  \end{align*}
\end{definition}
The $+$ and $\comp$ are the \emph{union} and \emph{composition}
operators. We call $S^=$ and $S^{\neq}$ as the \emph{$=$-restriction}
and the \emph{$\neq$-restriction} respectively.  From the definitions,
it is easy to see that $+$ is commutative and associative.
Furthermore, the operator $\circ$ is associative and distributes over
$+$.

For an REE $e$, let us write $S_e$ for the binary relation defined by
it: $ S_e = \{ \struct{u,v} \in V \times V~|~ \exists w \in \Ll(e)
\text{ s.t. } u \xra{w} v \}$.
% \begin{align*}
%   S_e = \{ \struct{u,v} \in V \times V~|~ \exists w \in \Ll(e)
%   \text{ s.t. } u \xra{w} v \}
% \end{align*}

% Constructing all definable sets We now divide the binary relations
% into different levels based on the operators defined above.

\begin{definition}
  Consider the set $\bin$ of binary relations. We will define a
  sequence $L_0, L_1, \dots$ of subsets of $\bin$, called
  \emph{levels}, as follows:
  \begin{align*}
    B_0 =~ & \{S_\e\} ~\cup ~\{S_a~|~a \in \Sigma \} \\
    L_0 =~ & \text{closure of $B_0$ under $+$ and $\comp$} \\
    \text{for $i \ge 1$, } B_i =~ & \{ S^=~|~ S \in L_{i-1} \} \cup
    \{S^{\neq}~|~ S \in L_{i-1} \}
    \cup L_{i-1} \\
    L_i =~ & \text{closure of $B_i$ under $+$ and $\comp$}
  \end{align*}
  
\end{definition}

Intuitively, the set $B_{0}$ consists of those relations that are
defined using the atomic expressions $\e$ and $a$. The set $L_0$
closes these relations under union and composition. The base sets
$B_1$ of the next level are formed by adding to $L_0$ the $=$ and
$\neq$-restrictions of relations in $L_0$. These are now closed under
union and composition to get the set $L_1$. This process continues. Of
course, this cannot continue beyond $2^{n^2}$ steps, which is the
total number of relations in $\bin$.
%We will now observe some
%properties of this hierarchy of relations.
%
%
%To obtain new relations in $L_i$, the $+$ and $\comp$ operators can
%act in many different ways on the relations in $B_i$. The following
%lemma gives a simple structure to the way union and composition can be
%used to yield new sets.  Let $n$ denote the number of nodes in $G$.
%
% \begin{lemma}\label{lem:form-of-relation-union-comp}
%   For each $i \ge 1$, every relation $S \in L_i$ which is not in $B_i$
%   is of the following form:
%   \begin{align*}
%     S & = T_1 + T_2 + \dots + T_m && \text{where } m \le n^2 \\
%     \text{ and each } \quad T_j & = (R_1 \comp R_2 \comp \dots \comp
%     R_p) &&
%     \text{s.t. } p \le 2^{n^2}, \text{ and } \\
%     &&& R_k \in B_i \text{ for $k \le p$. }
%   \end{align*}
% \end{lemma}
% $L_{n^2} = L_{n^2 + 1}$.
The next lemma further restricts it to $n^{2}$ steps.
\begin{lemma}
  \label{lem:polynomial-height}
  For all $j \ge n^2$, $L_j = L_{n^2}$.
\end{lemma}
\begin{proof}
  To every newly added relation in $L_i$, we will associate a relation
  in $L_{i-1}$ having at least one extra pair of nodes. As the new
  relations that are added become strictly smaller each time we go up
  a level, no new relations can be added beyond $L_{n^2}$.

  A newly added relation $S$ in $L_i$ is either in $B_i$, or formed by
  union of compositions of relations from $B_i$:
   \begin{align}
    S & = T_1 + T_2 + \dots + T_m && \text{where } m \le n^2 \label{eq:1} \\
    \text{ and each } T_j & = (R_1 \comp R_2 \comp \dots \comp
    R_p) &&
    \text{s.t. } p \le 2^{n^2}, \text{ and } \nonumber\\
    &&& R_k \in B_i \text{ for $k \le p$. } \nonumber
  \end{align}
  The bounds on $m$ and $p$ above follow from the fact that each
  relation can have at most $n^2$ pairs.
  For the relations $S^=$ or $S^{\neq}$ added in $B_i$ let us
  associate the set $S \in L_{i-1}$. They are added only if
  they are strict subsets of $S$. This means there exists a pair
  $\struct{u,v} \in S$
  that does not belong to $S^=$ and another pair $\struct{u',v'} \in S$ that
  does not belong to $S^{\neq}$. Hence the cardinalities of $S^=$ and
  $S^{\neq}$ are strictly lesser than $S$.

  Let $T$ be a relation obtained by union of compositions of relations
  in $B_i$. The relation $T$ is new only if some of the relations in
  the underlying composition according to
  \eqref{eq:1} are the new basic sets
  $S^=$ and $S^{\neq}$. Consider the relation $T'$ where each of these
  $S^=$ and $S^{\neq}$ is replaced by the corresponding relation
  $S$. Clearly $T \incl T'$. Moreover $T$ is added only if it is
  different from $T'$. This shows that $T'$ has at least one pair more
  than $T$.
\end{proof}

% Let us now relate the above hierarchy of relations obtained by the
% action of operators to the relations defined by REE.

The motive behind defining these operations and the hierarchy of
relations is that this procedure resembles the way REEs are constructed
from its grammar.

% The crux is the following lemma.  The following lemma states an
% important observation that the above operators are compatible with
% the operations used in the construction of REE. The proof of the
% lemma is straightforward from the definitions.

% Main observation that composition is defined by composition of RE,
% same for union
\begin{lemma}\label{lem:ree-comp}
  For every two REE $e$ and $f$, we have: $S_e + S_f = S_{e+f}$, $S_e
  \comp S_f = S_{ef}$, $S_{e}^{=} = S_{e_=}$ and $S_{e}^{\neq} =
  S_{e_{\neq}}$.
\end{lemma}
The proof of the above lemma is quite straightforward from the
definitions. However, the lemma is significant because it allows
to reason about $S_{ef}$ by independently reasoning about
$S_{e}$ and $S_{f}$. As mentioned before, this property is not true for
REMs. The above lemma can be used to show the important property that
all REE-definable relations can be generated by this hierarchical
construction that repeatedly applies the $=$ and $\neq$ restrictions
and closes under $+$ and $\comp$.

% To show: A relation $S$ is REE-definable iff $S$ belongs to
% $L_{n^2}$.
\begin{lemma}\label{lem:REE-definable-in-hierarchy}
  A relation is $\rdpqe$-definable iff it belongs to level $L_{n^2}$.
\end{lemma}
\begin{proof}
  We prove the left-to-right direction by an induction on the
  structure of REE. Relations $S_\e$ and $S_a$ definable by the basic
  REE $\e$ and $a$ already belong to $L_0$, and hence belong to
  $L_{n^2}$.  For REE $e$ and $f$, let $S_e$ and $S_f$ belong to
  $L_{n^2}$.
  
  By Lemma~\ref{lem:ree-comp}, we have $S_{e+f} = S_e +
  S_f$ and $S_{ef} = S_e \comp S_f$. As $L_{n^2}$ is closed under $+$
  and $\comp$, the relations $S_{e+f}$ and $S_{ef}$ belong to
  $L_{n^2}$ as well. The relation $S_{e_=}$ equals $S_e^=$. Since
  $S_e$ belongs to $L_{n^2}$, the relation $S_e^=$ would be present in
  $B_{n^2+1}$ by definition and hence in the level $L_{n^2+1}$. But by
  Lemma~\ref{lem:polynomial-height}, this means that $S_e^=$ belongs
  to $L_{n^2}$ as well. Similar argument holds for $S_e^{\neq}$.

  The right-to-left direction can be proved by an easy induction on
  the level number, once again using Lemma~\ref{lem:ree-comp}.
\end{proof}

Let us define the \emph{height} of a relation $S$ to be the least $i$
such that $S \in L_i$. The fact that the height of an $\rdpqe$
definable relation is polynomially bounded can be used to give a
$\PSPACE$ upper bound.

% PSPACE algorithm

\begin{lemma}
  \label{prop:ree-upper-bound}
  $\rdpqe$ definability problem is in \PSPACE.
\end{lemma}

\begin{proof}
  \begin{figure}[t]
  \centering
    \begin{tikzpicture}[vertex/.style={rectangle, inner sep=2pt, fill=gray},scale=0.8]
      %\draw [help lines] (0,0) grid (8,7);

      \node (s) at (4, 7) {\scriptsize $(S,h)$};
      \node (s1) at (4, 6.3) {\scriptsize $+$};
      \node (t1) at (1, 5.5) {\scriptsize $(T_1, h_1)$};
      \node  (ti) at (4, 5.5) {\scriptsize $(T_i, h_i)$};
      \node  (tm) at (7, 5.5) {\scriptsize $(T_m, h_m)$};
      \node (tic) at (4, 4.7) {$\comp$};
      \node (r1) at (3,4) {\scriptsize $(R_1, h_{i_1})$};
      \node (rp) at (6,4) {\scriptsize $(R_2 \comp \dots \comp R_p,
        h_{i_2}, p-1)$};
      \node (eq) at (3,3.3) {$=$};
      \node (x) at (3, 2.5) {\scriptsize $(X, h_{i_1} - 1)$};
      \node (b) at (3,0.6) {\scriptsize $(S_a, 0)$};

      \node at (2.5, 5.5) {$\dots$};
      \node at (5.5, 5.5) {$\dots$};

      \begin{scope}[thin]
      \draw [thin] (s) -- (s1);
      \draw [thin] (s1) -- (t1);
      \draw [thin] (s1) -- (ti);
      \draw [thin] (s1) -- (tm);
      \draw [thin] (ti) -- (tic);
      \draw (tic) -- (r1);
      \draw (tic) -- (rp);
      \draw (r1) -- (eq);
      \draw (eq) -- (x);
      \draw [decorate, decoration=snake] (x) -- (3,1);
      \end{scope}
    \end{tikzpicture}
    \caption{A part of execution of the algorithm for $\rdpqe$ definability}
\label{fig:ree-proof}
  \end{figure}
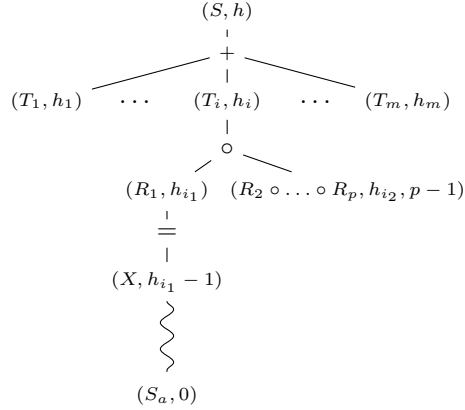

%%% Local Variables: 
%%% mode: latex
%%% TeX-master: "main"
%%% End: 

  The inputs are a data graph $G$ and a binary relation $S$. We will
  describe a non-deterministic algorithm that decides in polynomial
  space if $S$ is $\rdpqe$-definable. Due to
  Lemma~\ref{lem:REE-definable-in-hierarchy}, checking if $S$ is
  $\rdpqe$ definable is equivalent to checking if $S$ belongs to
  $L_{n^2}$. We will now explain how the algorithm can check the
  membership of $S$ in $L_{n^2}$.
  
  Every relation in $L_{n^2}$ has appeared due to a sequence of
  computations starting from the basic relations $S_\e$ and $S_a$ for
  each $a \in \S$. Thanks to
  \eqref{eq:1}, there is a specific
  structure to this computation that allows to look at it as a
  ``computation tree''. Nodes in this tree are relations in $L_{n^2}$
  and the children of a node are different relations in the same or
  smaller level that are used to construct the parent relation through
  the operations $+$, $\comp$, or the $=$ and $\neq$ restrictions. The
  leaves are the basic sets $S_\e$ and $S_a$. % The height
  % of this tree is polynomial. But the width could be exponential as
  % Lemma~\ref{lem:form-of-relation-union-comp} says that up to
  % $2^{n^2}$
  % relations could be composed to give a new relation.

  A naive approach would be to guess this entire computation tree and
  check two things: are leaves of the form $S_a$ or $S_\e$, and is
  each node either the union, composition or one of the (in)equality
  restrictions of its children. These checks can be done in polynomial
  time. However, we cannot hope to maintain the entire tree in
  \PSPACE{} as compositions can have exponentially many children
  \eqref{eq:1}.

  Instead of guessing the entire tree, the algorithm guesses, in a
  certain way, a path in the tree along with the children of each node
  in this path: if the node is a union of $T_1, \dots, T_m$, all these
  relations are guessed as children (there are only polynomially
  many); if a node is a composition of $R_1, \dots, R_p$, two children
  are guessed - the relation $R_1$ and the composed relation $R_2
  \comp \cdots \comp R_p$. The former is a proper child in the
  computation tree, and the latter is a different relation which needs
  to be further decomposed to get the actual children in the
  computation tree. The algorithm also maintains the number the
  decompositions left: so the child would be $(R_2 \comp \cdots \comp
  R_p, p-1)$. Storing $p-1$ needs only polynomially many bits. This is
  the main idea.  Additionally, the algorithm maintains the height of
  each node (see Figure~\ref{fig:ree-proof}). Each time an $=$ or
  $\neq$ restriction happens, the height reduces by one. At the leaf
  level, it is checked if the relation is one of $S_\e$ or
  $S_a$. Nodes whose heights have been certified can be removed: that
  is, leaf nodes can be removed after the basic check, and non-leaf
  nodes are removed once all its children are checked.  Hence at any
  point of time, the algorithm maintains some structure like
  Figure~\ref{fig:ree-proof}. Since the height is polynomial, this can
  be done in \PSPACE{}.
\end{proof}

The PSPACE{}-hardness of RPQ-definability~\cite{ANS2013} can be easily
extended to give the same lower bound for $\rdpqe$-definability.

\begin{theorem}\label{thm:ree-pspace-complete}
  $\rdpqe$-definability is $\PSPACE$-complete.
\end{theorem}
\begin{proof}
  The $\PSPACE$ upper bound comes from
  Lemma~\ref{prop:ree-upper-bound}. For the lower bound,
  consider the RPQ-definability problem: given a graph $H$ and a
  relation $T$ on $H$, is $T$ definable by a regular expression? This
  problem is known to be $\PSPACE$-complete~\cite{ANS2013}. We will
  show that RPQ-definability can be reduced to $\rdpqe$ definability.

  Consider the data graph $H'$ obtained from $H$ by attaching the same data value to
  all the nodes. Consider the
  $\rdpqe$ definability of the same set $T$. If $T$ is RPQ-definable on $H$,
  then clearly $T$ would be $\rdpqe$ definable on $H'$ as well using the
  same expression. Suppose $T$ is $\rdpqe$ definable on $H'$ using the
  expression $e$. % If $T$ is empty, then it would be RE-definable on
  % $H$ as well.
  Without loss of generality, assume that $T$ is non-empty. We claim that the REE
  defining $T$ on $H'$ will not have sub-expressions of the form
  $f_{\neq}$. This is because $f_{\neq}$ defines the empty relation
  and hence can be eliminated from the defining REE. Secondly observe
  that for the graph $H'$, no matter which REE $f$ we choose, $S_{f} =
  S_{f_=}$. Hence all sub-expressions of the form $f_{=}$ in $e$ can be
  modified to $f$ and still the defined set remains the same. This
  way, we have obtained a regular expression that defines $H'$. The
  same regular expression would define $T$ on $H$ as well.

\end{proof}

%%% Local Variables: 
%%% mode: latex
%%% TeX-master: "main"
%%% End: 

\section{Union of Conjunctive Queries}
\label{sec:ucq}
In this section, we study the complexity of the definability problem
for data graphs using UCRDPQs. The notion of homomorphisms has been
used in relational databases to characterize relations definable by
union of conjunctive queries. We will now adapt it to data graphs.

\begin{definition}
    \label{def:dataHomomorphism}
    Let $G = (V, E, \rho)$ be a data graph and $h: V \to V$ be a
    mapping. We call $h$ a \emph{data graph homomorphism} if it
    satisfies the following two conditions.
    \begin{enumerate}
        \item (Single step compatibility) For every $p, q \in V$ and
            $a \in \Sigma$, $p \xrightarrow{a} q$ implies $h(p)
            \xrightarrow{a} h(q)$.
        \item (Data compatibility of reachable nodes) For every
            $p, q \in V$, if $q$ is reachable from $p$, then
            $\rho(p) = \rho(q) \Leftrightarrow \rho(h(p)) =
            \rho(h(q))$.
    \end{enumerate}
\end{definition}
Intuitively, a data graph homomorphism $h$ ensures that if there is an
edge labeled $a$ from $p$ to $q$, there is also an edge labeled $a$
from $h(p)$ to $h(q)$, thus preserving the relations induced by the
letters in the finite alphabet. In addition, suppose there is a path
from $p$ to $q$. Then the data values at $p$ and $q$ are same if, and
only if, the data values at $h(p)$ and $h(q)$ are same. This preserves
the relations induced by (in)equality of data values. The following
result characterizes UCRDPQ-definable sets in terms of data graph
homomorphisms.

\begin{lemma}
    \label{lem:definableSetsHomomorphisms}
    Let $G = (V, E, \rho)$ be a data graph and $S$ be a relation of
    any arity. Then the following are equivalent.
    \begin{enumerate}
        \item The set $S$ is UCRDPQ-definable.
        \item For every data graph homomorphism $h$ and every tuple
            $\bar{p} \in S$, $h(\bar{p})$ also belongs
            to $S$.
    \end{enumerate}
\end{lemma}
\begin{proof}
    (1 $\Rightarrow$ 2). Suppose $S = Q_{1}(G) \cup \cdots \cup
    Q_{k}(G)$, where $Q_{1}, \ldots, Q_{k}$ are CRDPQs. Suppose
    $Q_{j}$ is of the form \eqref{eq:crdpq}(in~\defref{def:cdpq}) and
    $\bar{p} \in Q_{j}(G)$. Let $h$ be a data graph homomorphism. We
    will prove that $h(\bar{p}) \in Q_{j}(G)$. From the semantics of
    CRDPQs, we infer that there is a valuation $\mu: \cup_{1 \le i \le
    m} \{x_{i}, y_{i}\} \to V$ such that $(G, \mu) \models Q_{j}$ and
    $\bar{p} = \mu(\bar{z})$. Let $h \circ \mu$ be the valuation such
    that $h \circ \mu(x) = h(\mu(x))$ for all $x$. It is enough to
    prove that $(G, h \circ \mu) \models
    Q_{j}$ --- in that case,
    $h(\bar{p}) = h \circ \mu(\bar{z}) \in Q_{j}(G)$. Now let us prove that
    $(G, h \circ \mu) \models Q_{j}$. Since $Q_{j}$ is of the form
    \eqref{eq:crdpq} and $(G, \mu) \models Q_{j}$, we infer that for
    every $i= 1, \ldots, m$, there is a data path $w_{i} \in
    \Ll(e_{i})$ from $\mu(x_{i})$ to $\mu(y_{i})$. It is enough to prove that
    there is a data path $w_{i}'$ automorphic to $w_{i}$
    from $h \circ \mu(x_{i})$ to $h \circ \mu(y_{i})$ for every
    $i = 1, \ldots , m$ --- in that case $(G, h \circ \mu) \models
    Q_{j}$. So suppose $w_{i}$ is the data path associated with the
    path $ \mu(x_{i}) \xrightarrow{a_{1}} p_{2} \xrightarrow{a_{2}} \cdots
    \xrightarrow{a_{l-1}} p_{l} \xrightarrow{a_{l}} \mu(y_{i})$, where
    $p_{2}, \ldots, p_{l}$ are the intermediate nodes in the path from
    $\mu(x_{i}) = p_{1}$ to $\mu(y_{i})=p_{l+1}$. Since $h$ is a data
    graph homomorphism, we infer from the single step compatibility
    property that $G$ has a path $h \circ \mu(x_{i})
    \xrightarrow{a_{1}} h(p_{2}) \xrightarrow{a_{2}} \cdots
    \xrightarrow{a_{l-1}} h(p_{l}) \xrightarrow{a_{l}} h \circ
    \mu(y_{i})$, where $h(p_{2}), \ldots, h(p_{l})$ are intermediate
    nodes in a path from $h \circ \mu(x_{i})$ to $h \circ \mu(y_{i})$.
    We claim that the data path $w_{i}'$ associated with this path is
    automorphic to $w_{i}$.  If not, there would be positions $j_{1}$
    and $j_{2}$ such that $\rho(p_{j_{1}}) = \rho(p_{j_{2}})$ but
    $\rho(h(p_{j_{1}})) \ne \rho(h(p_{j_{2}}))$ (or vice-versa),
    violating the data compatibility property of the data graph
    homomorphism $h$. Thus, $w_{i}'$ is a data path automorphic to
    $w_{i}$ from $h \circ \mu(x_{i})$ to $h \circ \mu(y_{i})$. This
    concludes the proof that $h(\bar{p}) \in S$.

    (2 $\Rightarrow$ 1) Suppose $V = \{p_{1}, \ldots, p_{n} \}$. Let
    $\bar{x} = \struct{x_{1}, \ldots, x_{n}}$. Let
    $\phi_{G}(\bar{x})$ be defined as follows.
    \begin{align*}
        \phi_{G}(\bar{x}) & = \bigwedge_{(p_{i}, a, p_{j}) \in E}
        x_{i} \xrightarrow{a}x_{j} \quad \land \bigwedge_{(p_{i},
        p_{j}) \in (\Sigma^{+})_{=}(G)} x_{i}
        \xrightarrow{(\Sigma^{+})_{=}} x_{j}\\
        &\land  \bigwedge_{(p_{i}, p_{j}) \in (\Sigma^{+})_{\ne}(G)}
        x_{i} \xrightarrow{(\Sigma^{+})_{\ne}} x_{j}
    \end{align*}
    In the above definition, $(\Sigma^{+})_{=}$ is an REE;
    $(\Sigma^{+})_{=}(G)$ is the set of pairs of nodes $(p_{i},
    p_{j})$ such that $p_{j}$ is reachable from $p_{i}$ and both nodes
    have the same data value. Similarly, $(\Sigma^{+})_{\ne}(G)$ is
    the set of pairs of nodes $(p_{i}, p_{j})$ such that $p_{j}$ is
    reachable from $p_{i}$ and the two nodes have different data
    values. The valuation that assigns $p_{i}$ to $x_{i}$ for every
    $i=1, \ldots, n$ satisfies all the conditions in
    $\phi_{G}(\bar{x})$. If a valuation $\mu$ for $\bar{x}$ satisfies
    all the conditions in $\phi_{G}(\bar{x})$, then the mapping
    $h_{\mu}: V \to V$ such that $h_{\mu}(p_{i}) = \mu(x_{i})$ is a data
    graph homomorphism. Let $S'$ be the set of tuples defined by the
    UCRDPQ $\{\mathit{Ans}(\struct{x_{i_{1}}, \ldots, x_{i_{r}}
    }) := \phi_{G}(\bar{x}) \mid \struct{p_{i_{1}}, \ldots,
    p_{i_{r}} } \in S\}$. We claim that $S' = S$, which will
    prove that $S$ is UCRDPQ-definable.

    ($S \subseteq S'$): For every $\struct{p_{i_{1}}, \ldots, p_{i_{r}} }
    \in S$, we have that $\struct{p_{i_{1}}, \ldots, p_{i_{r}} } \in
    \mathit{Ans}(\struct{x_{i_{1}}, \ldots, x_{i_{r}} }) :=
    \phi_{G}(\bar{x})(G) \subseteq S'$.

    ($S' \subseteq S$): Suppose $\struct{p_{j_{1}}, \ldots, p_{j_{r}}
    } \in S'$. Then there is some $\struct{p_{i_{1}}, \ldots, p_{i_{r}} }
    \in S$ and a valuation $\mu$ for $\bar{x}$ such that $\mu$ satisfies
    all the conditions in $\phi_{G}(\bar{x})$ and $\struct{p_{j_{1}},
    \ldots, p_{j_{r}} } = \struct{\ \mu(x_{i_{1}}), \ldots, \mu(x_{i_{r}})
    }$. The mapping $h_{\mu}: V \to V$ such that $h_{\mu}(p_{i}) =
    \mu(x_{i})$ is a data graph homomorphism. Now we have
    $\struct{p_{j_{1}}, \ldots, p_{j_{r}} } = \struct{\mu(x_{i_{1}}),
    \ldots, \mu(x_{i_{r}})}$ and in addition $\struct{\mu(x_{i_{1}}), \ldots,
    \mu(x_{i_{r}})} = \struct{h_{\mu}(p_{i_{1}}), \ldots,
    h_{\mu}(p_{i_{r}})}$. Hence, we have $\struct{ p_{j_{1}}, \ldots, p_{j_{r}} } =
    h_{\mu} (\struct{ p_{i_{1}}, \ldots, p_{i_{r}}}) \in S$; the last
    inclusion follows from condition 2 of the lemma, as
    $\struct{p_{i_{1}}, \ldots, p_{i_{r}} } \in S$.  \end{proof}

Readers familiar with Global as View (GAV) schema mappings for
relational databases may note some similarities with the above result.
For a data graph $G$, consider the relational database
$D_{G}$ over the domain $V$ consisting of all the binary relations
that are definable by RDPQs. For a set of tuples $S$ over $V$, let
$D_{S}$ be the relational database consisting of the single relation
$S$. Then $S$ is UCRDPQ-definable on $G$ iff some GAV schema
mapping fits the source database $D_{G}$ and the target database
$D_{S}$. A characterization using homomorphisms similar to the one in
\lemref{lem:definableSetsHomomorphisms} is given for GAV schema
mappings in \cite{CKT2010, ATKT2011}. Here, we extend the notion of
homomorphisms to include data value compatibility. A
\coNP{}-completeness result for a subclass of GAV schema mappings is
given in \cite{CKT2010}. We give a similar result for
UCRDPQ-definability below. However, there is no obvious way of
directly using the upper bound in \cite{CKT2010} here, since the
relational database $D_{G}$ may be exponentially larger than $G$.
Preservation under homomorphism is a fundamental concept, which
appears in other contexts as well, for example querying databases with
incomplete information \cite{GheerbrantLS13}.

\begin{theorem}
    \label{thm:ucrdpqCoNP}
    UCRDPQ-definability is \coNP{}-complete.
\end{theorem}
\begin{proof}
    We first prove the \coNP{} upper bound. Given a data graph
    $G = (V, E, \rho)$ and a set of tuples $S$, we can guess a mapping
    $h: V \to V$, verify that it is a data graph homomorphism and that
    there is some tuple $\bar{p} \in S$ such that $h(\bar{p}) \notin
    S$. If $S$ is not UCRDPQ-definable, then
    \lemref{lem:definableSetsHomomorphisms} ensures that at least one
    of the guesses will succeed. On the other hand, if $S$ is
    UCRDPQ-definable, then none of the guesses will succeed.

    For the \coNP{} lower bound, we reduce the unsatisfiability
    problem for Boolean 3-CNF formulas to the UCRDPQ-definability
    problem. This part of the proof is an adaptation of a similar
    proof from \cite{CKT2010} about relational databases (which can
    have ternary relations) to data graphs (which can only have binary
    relations). Given a 3-CNF formula $F$ consisting of clauses
    $C_{1}, \ldots, C_{m}$ over the variables $p_{1}, \ldots, p_{n}$,
    we map it to the data graph shown in \figref{fig:ucrdpqLowBound}.
    All nodes have the same data value, which is not shown explicitly.
    \begin{figure*}
        \centering
        \input{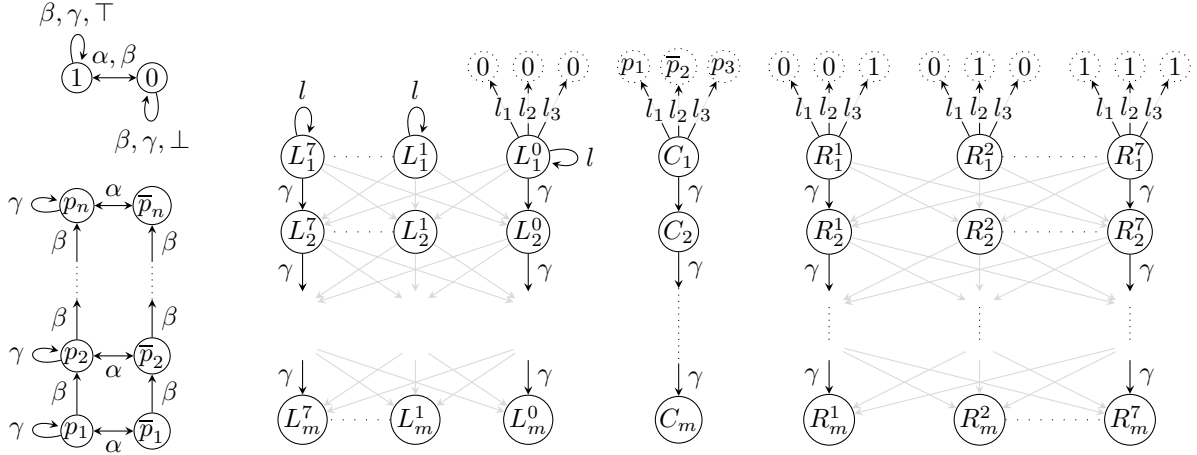}
        \caption{Data graph for the lower bound of
        UCRDPQ-definability}
        \label{fig:ucrdpqLowBound}
    \end{figure*}

    There is only one node {\tikz \node[state] {$1$};}, but the same
    node is drawn (dotted) at many places in the diagram to avoid the
    clutter of too many edges. Same comment applies to all the nodes
    drawn dotted. There is an edge labeled $\g$ from all nodes
    $R_{i}^{j}$ to $R_{i+1}^{k}$, but most of the edges are grayed out
    and the label $\g$ is not shown to reduce clutter. Same applies to
    edges from $L_{i}^{j}$ to $L_{i+1}^{k}$. In the diagram, the
    clause $C_{1}$ is assumed to be $(p_{1} \lor \lnot p_{2} \lor
    p_{3})$.  From every clause node $C_{i}$, there are edges labeled
    $l_{1}$, $l_{2}$ and $l_{3}$ to the nodes corresponding to the
    literals occurring in $C_{i}$. Only the edges from $C_{1}$ are
    shown and others are not shown. From every node $R_{i}^{j}$ and
    $L_{i}^{j}$, there is an edge labeled $l_{k}$ to either $0$ or
    $1$, depending on the $k$\textsuperscript{th} most significant bit
    of the binary representation of $j$. Only the edges from
    $L_{1}^{0}$, $R_{1}^{1}$, $R_{1}^{2}$ and $R_{1}^{7}$ are shown.
    Others are not shown. We claim that the given Boolean 3-CNF
    formula $F$ is not satisfiable iff the set of tuples of nodes $S =
    \{\struct{C_{1}}, \ldots, \struct{C_{m}}\} \cup
    \{\struct{L_{i}^{j}} \mid 1 \le i \le m, 0 \le j \le 7\}$ is
    UCRDPQ-definable.

    Suppose there is an assignment $\mathit{sa}: \{p_{1},
    \ldots, p_{n}\} \to \{0, 1\}$ satisfying $F$.  Consider the graph
    mapping $h$ that maps the node $p_{i}$ to the node
    $\mathit{sa}(p_{i})$ and $\bar{p}_{i}$ to $1 -
    \mathit{sa}(p_{i})$. For every $i=1, \ldots, m$, $h$ maps the node
    $C_{i}$ to the node $R_{i}^{j}$; here $j$ is the number whose
    binary representation is the one formed by the three literals of
    the clause $C_{i}$ according to the satisfying assignment
    $\mathit{sa}$. All other nodes are mapped to themselves by
    $h$. This mapping $h$ is a data graph homomorphism and
    $h(\struct{C_{1}}) = \struct{R_{1}^{j}}$ for some
    $j$. Since $\struct{C_{1}} \in S$ and $\struct{R_{1}^{j}} \notin
    S$, we infer from \lemref{lem:definableSetsHomomorphisms} that $S$
    is not UCRDPQ-definable.

    Conversely, suppose $F$ is not satisfiable. Let $h$ be any data
    graph homomorphism. We will prove that $h(\struct{p}) \in S$ for
    every tuple $\struct{p} \in S$. Since the only node with a self
    edge labeled $\top$ (resp.~$\bot$) is $1$ (resp.~$0$), $1$
    (resp.~$0$) is mapped to itself by $h$. Due to the self edges
    labeled $l$ and the edges labeled $l_{1}$, $l_{2}$ and $l_{3}$,
    $h$ maps $L_{1}^{j}$ to itself for every $j=0, \ldots, 7$. The
    edges labeled $\g, l_{1}, l_{2}, l_{3}$ then force $h$ to map $L_{i}^{j}$ to itself
    for every $i,j$. It remains to prove that $h(\struct{C_{i}}) \in
    S$ for every $i = 1, \ldots, m$. Due to the edges labeled
    $\b$ and the self edges labeled $\g$, $h$ maps $p_{1}$ to either
    itself or to $1$ or to $0$. If $h$ maps $p_{1}$ to itself, then
    the edges labeled $\a$ and $\b$ force $h$ to map $p_{i}$ to itself
    (and $\bar{p}_{i}$ to itself) for every $i = 1, \ldots, n$. The
    edges labeled $l_{1}$, $l_{2}$ and $l_{3}$ then force $h$ to map
    $C_{i}$ to itself for every $i=1, \ldots, m$. On the other hand,
    if $h$ maps $p_{1}$ to $1$ or $0$, the edges labeled $\a$ force
    $h$ to map $\bar{p}_{1}$ to $1 - h(p_{1})$. The edges labeled $\a$
    and $\b$ then force $h$ to map $p_{i}$ to $1$ or $0$ and
    $\bar{p}_{i}$ to $1 - h(p_{i})$ for every $i = 1, \ldots, n$. The
    homomorphism $h$ thus determines a truth assignment for $p_{1},
    \ldots, p_{n}$. For every $i=1, \ldots, m$, the edges labeled
    $l_{1}$, $l_{2}$ and $l_{3}$ force $h$ to map $C_{i}$ to
    either $L_{i}^{j}$ or $R_{i}^{j}$; here $j$ is the number whose
    binary representation is the one formed by the three literals of
    the clause $C_{i}$ according to the truth assignment determined by
    $h$. If $h$ maps $C_{i}$ to $R_{i}^{*}$ ($R_{i}^{*}$ could be any
    one of $R_{i}^{1}$, \ldots, $R_{i}^{7}$) for some $i = 1, \ldots,
    m$, then the edges labeled $\g$ force $h$ to map $C_{i}$ to
    $R_{i}^{*}$ for every $i = 1, \ldots, m$. This implies that the
    truth assignment determined by $h$ assigns at least one literal to
    $\mathit{true}$ in every clause, contradicting the hypothesis that
    $F$ is not satisfiable. Hence $h$ maps $C_{i}$ to $L_{i}^{*}$ for
    every $i = 1, \ldots, m$. Since this holds for every data graph
    homomorphism, we conclude that $h(\struct{p}) \in S$ for every
    data graph homomorphism $h$ and every tuple $\struct{p} \in S$.
    Hence, we can conclude from
    \lemref{lem:definableSetsHomomorphisms} that $S$ is
    UCRDPQ-definable.
\end{proof}

%%% Local Variables: 
%%% mode: latex
%%% TeX-master: "main"
%%% End: 

\section{Discussion}

A natural question to ask is how to synthesize a query that defines a
given relation. In principle, the decision procedures in the paper can
be converted into a procedure to synthesize a defining query. However
such queries would not have an interesting structure. For instance, in
the REM and REE cases, the synthesized queries do not make use of the
star operator.  Moreover, the lower bound for the decision problem
implies that the worst case size of the defining queries will be
doubly exponential for REMs, and exponential for REEs.  In the UCRDPQ
case, the defining query described in
Lemma~\ref{lem:definableSetsHomomorphisms} essentially constructs the
whole data graph using conditions and then picks out the required
tuples. This does not capture the essence of conjunctive queries,
which is to identify patterns that are much smaller than the graphs
themselves.

A possible future direction would be to find a notion of ``good''
queries and reformulate the definability problem to ask for the
existence of ``good'' defining queries. 

In some application domains, data graphs may have a special structure
(such as not too many cycles). An orthogonal direction would be to
study the definability problem for such data graphs.

%%% Local Variables: 
%%% mode: latex
%%% TeX-master: "main"
%%% End: 

\bibliographystyle{plain}
\bibliography{references}

\end{document}